\documentclass[journal,twoside,web]{ieeecolor}
\usepackage{generic}
\usepackage{cite}
\usepackage{amsmath,amssymb,amsfonts,bm}
\usepackage{algorithmic}
\usepackage{graphicx}
\usepackage{textcomp}
\usepackage{caption}
\usepackage{subcaption}
\usepackage{mwe}
\usepackage{cite}
\usepackage{amsmath}
\usepackage{amssymb,amsfonts,bm}
\usepackage{algorithmic}
\usepackage{graphicx}
\usepackage{textcomp}
\usepackage{mathtools}
\usepackage{booktabs}
\usepackage{multicol}

\usepackage{hyperref}

\usepackage{textcomp}
\usepackage{hyperref}
\usepackage{eso-pic}

\usepackage{tikz}
\usetikzlibrary{positioning, arrows.meta}

\AddToShipoutPicture{\begin{tikzpicture}[remember picture,overlay]
		\node[anchor=north,yshift=-10pt,align=center] at (current page.north) {\footnotesize This is the authors' electronic preprint version of an article submitted to IEEE for publication.};%\\ \footnotesize The final published paper can be found at \href{https://doi.org/10.1109/TRO.2024.3366818}{DOI: 10.1109/TRO.2024.3366818}};% \footnotesize Copyright may be transferred without notice, after which this version may no longer be accessible.};% \\ 
\end{tikzpicture}}

\usepackage{xcolor}

\newtheorem{theorem}{\bf{Theorem}}
\newtheorem{lemma}{\bf{Lemma}}
\newtheorem{corollary}{\bf{Corollary}}
\newtheorem{assumption}{\bf{Assumption}}
\newtheorem{objective}{\bf{Objective}}

\newtheorem{remark}{\bf{Remark}}

\renewenvironment{proof}{{\bf{Proof:}}}{$\hfill\blacksquare$}
\newtheorem{proof_th}{\bf{Proof of Theorem}}

\newcommand{\fin}[1]{\textcolor{black}{#1}}%red}{#1}}

\def\BibTeX{{\rm B\kern-.05em{\sc i\kern-.025em b}\kern-.08em
T\kern-.1667em\lower.7ex\hbox{E}\kern-.125emX}}
\begin{document}

\newcommand\copyrighttext{%
\footnotesize \textcopyright 2024 IEEE. Personal use of this material is permitted. Permission
from IEEE must be obtained for all other uses, in any current or future
media, including reprinting/republishing this material for advertising or
promotional purposes, creating new collective works, for resale or
redistribution to servers or lists, or reuse of any copyrighted
component of this work in other works.}
\newcommand\copyrightnotice{%
\begin{tikzpicture}[remember picture,overlay]
\node[anchor=south,yshift=5pt] at (current page.south) {\fbox{\parbox{\dimexpr\textwidth-\fboxsep-\fboxrule\relax}{\copyrighttext}}};
\end{tikzpicture}%
}

\title{Exponentially Stable Projector-based Control of Lagrangian Systems with Gaussian Processes}
\author{Giulio Evangelisti, Cosimo Della Santina, \IEEEmembership{Senior Member, IEEE}, and Sandra Hirche, \IEEEmembership{Fellow, IEEE} %
\thanks{Manuscript submitted for review on 19 May 2024. This work was supported by the Consolidator Grant “Safe data-driven control for human-centric systems” (CO-MAN) of the European Research Council (ERC) of the European Union (EU) under Grant 864686. (Corresponding author: Giulio Evangelisti.)}
\thanks{G. Evangelisti and S. Hirche are with the Chair of Information-oriented Control, TUM School of Computation, Information and Technology, Technical University of Munich, 80333 Munich, Germany (e-mail: {g.evangelisti; hirche}@tum.de).}
\thanks{C. Della Santina is with the Cognitive Robotics department, Delft University of Technology, Mekelweg 2, 2628 CD Delft, Netherlands (e-mail: c.dellasantina@tudelft.nl).}
}

\maketitle

\copyrightnotice

\begin{abstract}
Designing accurate yet robust tracking controllers with tight performance guarantees for Lagrangian systems is challenging due to nonlinear modeling uncertainties and conservative stability criteria. This article proposes a structure-preserving projector-based tracking control law for uncertain Euler-Lagrange (EL) systems using physically consistent Lagrangian-Gaussian Processes (L-GPs). We leverage the uncertainty quantification of the L-GP for adaptive feedforward-feedback balancing. In particular, an accurate probabilistic guarantee for exponential stability is derived by leveraging matrix analysis results and contraction theory, where the benefit of the proposed controller is proven and shown in the closed-form expressions for convergence rate and radius. Extensive numerical simulations not only demonstrate the controller's efficacy based on a two-link and a soft robotic manipulator but also all theoretical results are explicitly analyzed and validated.
\end{abstract}


\section{Introduction}
\label{sec:introduction}
\IEEEPARstart{D}{ata-driven} methods offer a promising approach to enhance the performance, reliability, and safety of traditional control strategies based on parametric mathematical models. Gaussian Processes (GPs) stand out among these methods due to their data-efficiency and uncertainty quantification. However, neither GPs nor other data-driven methods generally account for physical consistency~\cite{my_tro_article}, limiting their applicability in model-based control and the accuracy of providable guarantees.

Physics-informed machine learning has been shown to improve the methods' data efficiency and reliability by encoding the variational Euler-Lagrange (EL) structure into data-driven models for mechanical and electromagnetic systems \cite{geist_survey,symp_GPR_ham_sys,blobaum_math_journ_lsi}. In particular, deep learning with neural networks has shown auspicious results for robotic systems~\cite{lutter2021combining}. However, these methods typically lack a measure for uncertainty quantification, hampering the provision of safety guarantees, such as in, e.g., \cite{lutter_iros_ua_enrgy}, where Deep Lagrangian Networks (DeLaNs) and their passivity properties are leveraged for a learning-based implementation of an energy-based controller \cite{spong_ifac_energy} for under-actuated systems. Here, sufficiently low modeling errors are an inevitable requirement, which can only be satisfied in regions near the training domain. 

\begin{figure}
\centering
\includegraphics[width=\linewidth]{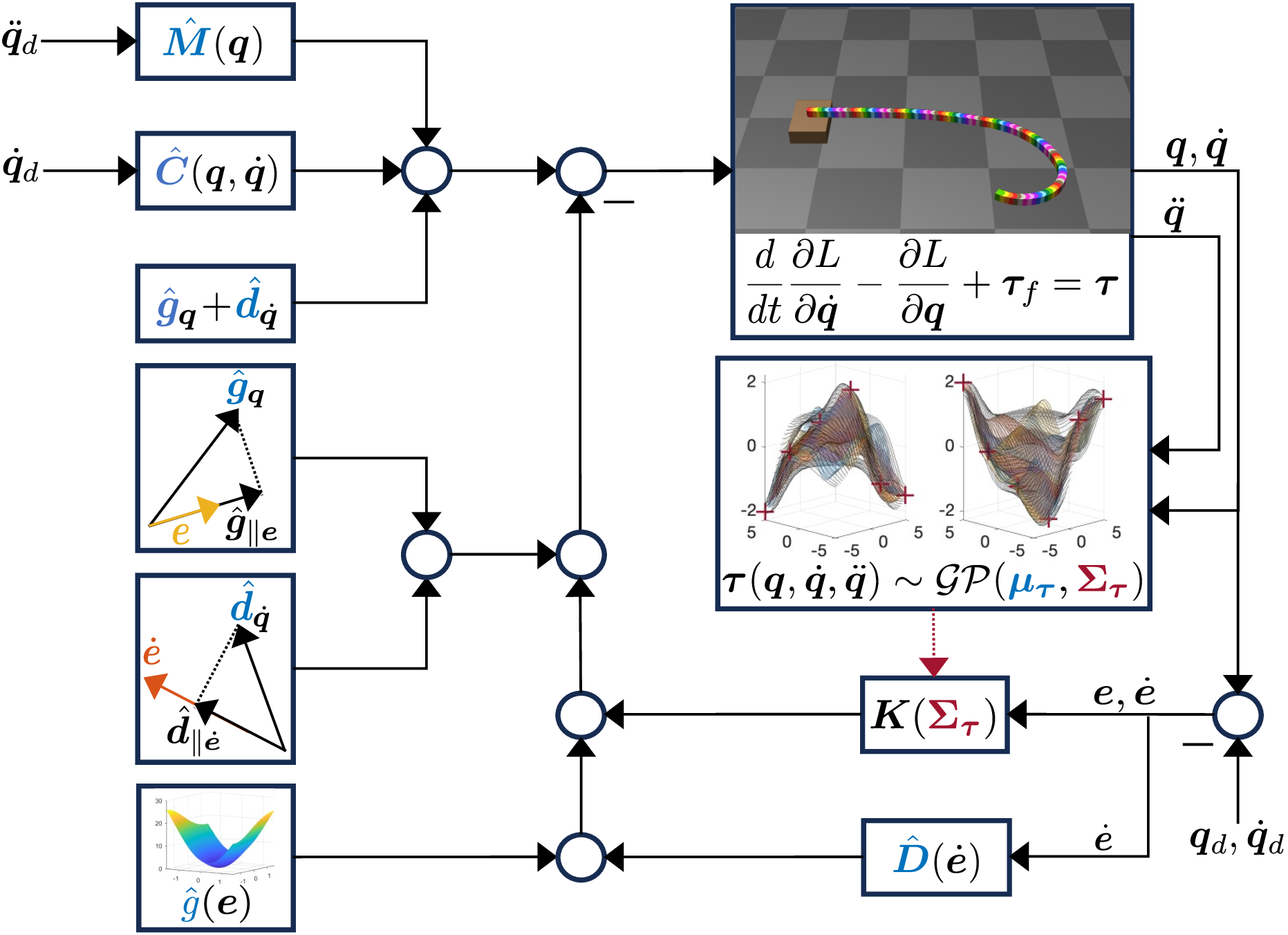}
\caption{Block scheme of the proposed projector-based L-GP control with uncertainty-adaptive feedback: $\hat{\bm{M}}$, $\hat{\bm{C}}$, $\hat{\bm{g}}_{\bm{q}}$ and $\hat{\bm{d}}_{\dot{\bm{q}}}$ denote the learned inertia, Coriolis matrix, gravity and friction, respectively, of the system with uncertain Lagrangian $L$.} 
\label{cover_fig}
\end{figure}

Applying GPs to learning-based control of robotic systems was initially proposed for local regression in a computed torque control scheme~\cite{tuong_iros}, yet without any stability guarantees. When safety is explicitly considered in data-driven control methods, guarantees often critically depend on the prior availability of a known dynamical model based on which, e.g., a reachability analysis for robotic systems with state-dependent disturbances can be performed~\cite{fisac_tac}. Feedback compensation of these residual dynamics based on black-box modeling with GPs is proposed in~\cite{thomas_tracking_control_of_el_systems} for EL systems together with an uncertainty-adaptive computed torque controller. Here, due to multiple applications of Young's inequality, the exponential convergence result suffers from conservatism and dependencies on several positive parameters that are assumed to exist. A GP-based feedback linearizing control law is formulated in \cite{jonas_fb_lin_ol} but is limited to single-input systems in controllable canonical form while asymptotic stability of the tracking error is only guaranteed in the noiseless case, reducing to ultimate boundedness otherwise. Performances of model predictive controllers have also been shown to improve significantly with learning via GPs~\cite{hewing_mpc}, where closed-loop stability guarantees have been proven for linear time-invariant systems~\cite{berberich_lti_tac} or input-to-state stability for nonlinear discrete-time systems~\cite{maiworm_mpc}. 

In the context of structure-preserving control, several different approaches have been proposed. An optimal control problem with respect to the external torque applied to a single rigid body is analyzed in \cite{bloch_opt}, but potential fields are assumed to be absent, and a numerical solution of optimality conditions up to the fourth order is required. A proportional-derivative controller is naturally generated in \cite{cosimo_ijrr_plan_soro} by leveraging the inherent impedance of the underlying soft robot. Therefore, its softness is preserved for potential environmental interactions, yet only linear stiffness and damping matrices are considered while gravity is fully compensated. Also contraction-based control design has gained a significant amount of research attention. However, it leads to the so-called integrability problem \cite{kawano_atmtc_2023} for nonlinear systems due to the involved Jacobians and, therein, partial derivatives with respect to the controllers. Altogether, leveraging physically consistent methods in learning-based control of robotic systems while being able to provide accurate stability guarantees for high performance and reliability is still an ongoing open problem. This article addresses these issues by leveraging the physical consistency of Lagrangian-Gaussian Processes (L-GPs) from~\cite{my_tro_article} in an uncertainty-adaptive control framework combined with rigorous theoretical guarantees.

\subsection{Contribution}

We propose a novel structure-preserving control design, which uses projectors to minimally alter the system's dynamics while leveraging the L-GP's covariance estimate for confidence-dependent gain adaptation. By preserving the system's dynamical structure, we aim to achieve a high level of robustness to remaining model uncertainties. At the same time, we exploit the learned model with projections to make use of the natural potential and dissipative forces driving towards the control goal, i.e., the desired trajectory in the state space. Moreover, a novel exponential stability result is derived based on intermediate matrix analysis results and arguments from contraction theory applied to Lagrangian systems. The integrability problem of contraction theory-based control due to the dependence on Jacobians is circumvented by directly exploiting the dynamics' structure. Analytically compact expressions for time-variant convergence radius and rate are derived with the goal of optimizing the tightness of the resulting bounds. At the same time, these stability results also significantly extend our earlier work, which combines the L-GP~\cite{mein_cdc_paper} with the standard PD+ controller~\cite{pdp_ref}, the latter is considered de-facto as the state of the art (SoA) in this work. All theoretical contributions are validated in numerical experiments, further underlining the practical efficacy of the proposed methods. Note that the L-GP modeling framework from~\cite{my_tro_article} is applied but does not represent a contribution of the present work.

The remainder of this article is organized as follows: After formally introducing our considered problem setting in section~\ref{sec_prob_form}, we first provide a brief overview of the modeling framework of L-GPs in section~\ref{gpr_sec}. The proposed control design is then derived and analyzed in section~\ref{sec_ctrl_dsgn}, followed by validating numerical simulations in section~\ref{num_sec}. Finally, we summarize our work's theoretical and practical implications in section~\ref{sec_conclusion}.

\subsection{Notation}

Bold lower and upper case symbols denote vectors $\bm{a}$ and matrices $\bm{A}$, $\bar{a}\!=\!\bar{\lambda}(\bm{A})$ and $\underline{a}\!=\!\underline{\lambda}(\bm{A})$ the maximal and minimal eigenvalues of $\bm{A}$, $\mathrm{E}[\cdot]$ and $\mathrm{Var}[\cdot]$ the expectation and variance operators, and $\mathbb{R}$ and $\mathbb{N}$ the set of real and natural numbers, respectively. $\bm{I}$ is the identity, $\bm{0}$ the zero and $\bm{1}$ the ones matrix. $|\cdot|$ indicates the cardinality of a set, $\mathcal{N}(\bm{\mu},\bm{\Sigma})$ a mul\fin{ti}variate Gaussian distribution with mean $\bm{\mu}$ and covariance $\bm{\Sigma}$, $\lVert\cdot\rVert_{\mathcal{L}_2}=\sqrt{\smallint_0^\infty\lVert\cdot\rVert^2dt}$ the $\mathcal{L}_2$-norm, and $\lVert\cdot\rVert$ the Euclidian norm if not stated otherwise. Positive definiteness (resp. semi-definiteness) of a symmetric matrix is indicated by $\cdot \succ \bm{0}$ (resp. $\cdot \succeq \bm{0}$) and $\mathrm{vec}(\cdot)$ stacks the columns of a matrix to form a vector. %Finally, $\gamma(n)$, $\gamma(n\!+\!1/2)$ and $\Gamma_n(\rho)$ denote the ordinary and incomplete (upper) gamma functions evaluated at positive and nonnegative integer $n$, respectively, see \cite{my_tro_article}.

\section{Problem Formulation}\label{sec_prob_form}

\subsection{Dynamical System}

In this article, we focus on uncertain, fully actuated Euler-Lagrange (EL) systems with equations of motion given by~\cite{ortega_pbc_el}
\begin{equation} % Lagrangian equations of motion
\frac{d}{dt} \frac{\partial L}{\partial\dot{\bm{q}}} - \frac{\partial L}{\partial\bm{q}} = \bm{\tau}_c = \bm{\tau} - \bm{\tau}_f \label{lagrange_fund_eqs_of_mot}
\end{equation}%nonconservative
with Lagrangian $L$, generalized coordinates $\bm{q}\in\mathbb{R}^N$, and generalized forces $\bm{\tau}_c\in\mathbb{R}^N$. The latter are composed of the difference between active input $\bm{\tau}$ and dissipative friction $\bm{\tau}_f$.

% Assumption #1: Autonomous Lagrangian, no dissipation
\begin{assumption}\label{ass_autonomous}
The Lagrangian $L\equiv L(\bm{q},\dot{\bm{q}}) = T(\bm{q},\dot{\bm{q}}) - V(\bm{q})$ in \eqref{lagrange_fund_eqs_of_mot} is autonomous and consists of the difference between unknown kinetic energy $T\colon\mathbb{R}^N\times\mathbb{R}^N\to\mathbb{R}$ and unknown potential energy $V\colon\mathbb{R}^N\to\mathbb{R}$.
\end{assumption}

% Assumption #2: Kinetic energy
\begin{assumption}\label{ass_kinetic}
The kinetic energy is quadratic $T(\bm{q},\dot{\bm{q}}) = \frac{1}{2}\dot{\bm{q}}^T\bm{M}(\bm{q})\dot{\bm{q}}$ w.r.t.~the velocities $\dot{\bm{q}}$, where $\bm{M}\colon\mathbb{R}^N\to\mathbb{R}^{N\times N}$ is the unknown, (symmetric) positive definite, inertia matrix. %generalized
\end{assumption}

% Assumption #3: Potential energy
\begin{assumption}\label{ass_potential}
The potential energy is positive-semidefinite and has an equilibrium in the origin, i.e., $V(\bm{0})=0$, $V(\bm{q})\geq0$ $\forall \bm{q}\in\mathbb{R}^N$, and $\tfrac{\partial}{\partial\bm{q}} V(\bm{0})=\bm{0}$.
\end{assumption}

% Assumption #4: Dissipation
\begin{assumption}\label{ass_disp}
The dissipative friction $\bm{\tau}_f\!\!=\!\!\bm{D}(\dot{\bm{q}})\dot{\bm{q}}$ is the matrix-vector product of the generalized velocities with the (symmetric) positive semi-definite damping matrix $\bm{D}(\dot{\bm{q}})\!\succeq\!\bm{0}$. 
\end{assumption}

Assumptions~\ref{ass_autonomous}--\ref{ass_disp} describe the considered system class and represent its physical properties. Note that Assumption~\ref{ass_autonomous} is nonrestrictive since time-variance and dissipation can be introduced to the conservative left-hand side of \eqref{lagrange_fund_eqs_of_mot} via the external force component $\bm{\tau}_f$. Assumption~\ref{ass_kinetic} enforces positivity of the kinetic or electric energy for all nonzeros velocities or currents, respectively, and is valid for, e.g., any non-relativistic mechanical system, while Assumption~\ref{ass_potential} requires w.l.o.g.~that the chosen coordinates have an equilibrium at the origin. Lastly, Assumption~\ref{ass_disp} assumes a certain friction or resistance structure applicable to various dissipation phenomena, e.g., linear viscous, air-drag, structural, or continuous Coulomb dampings. In case of a purely position- or charge-dependent damping matrix $\bm{D}(\bm{q})$, the Rayleigh dissipation potential $R(\bm{q},\dot{\bm{q}})\!\!=\!\!\tfrac{1}{2}\dot{\bm{q}}^\top\bm{D}(\bm{q})\dot{\bm{q}}$ structurally permits a derivation according to $\bm{\tau}_f\!=\!\tfrac{\partial R}{\partial\dot{\bm{q}}}\!=\!\bm{D}(\bm{q})\dot{\bm{q}}$.

\begin{remark}
For conciseness, the remainder of this work focuses on mechanical systems, cf. Assumption~\ref{ass_kinetic}. However, due to our variational modeling based on generalized energies, the system class \eqref{lagrange_fund_eqs_of_mot} naturally includes mechanical, electrical, and even mixed-nature, i.e., electromechanical systems \cite{ortega_pbc_el}.
\end{remark}

Exploiting Assumptions \ref{ass_autonomous}--\ref{ass_disp} and applying the chain rule to \eqref{lagrange_fund_eqs_of_mot}, we obtain the well-known matrix-vector expression
\begin{equation} % M-C-g like Matrix-Vector-Equation
\bm{M}(\bm{q})\ddot{\bm{q}} + \bm{C}(\bm{q},\dot{\bm{q}})\dot{\bm{q}} + \bm{g}(\bm{q}) + \bm{D}(\dot{\bm{q}})\dot{\bm{q}} = \bm{\tau} \, , \label{M_C_g_system}
\end{equation}
where $\bm{M}(\bm{q})\in\mathbb{R}^{N\times N}$ is the (symmetric) positive definite, inertia matrix, $\bm{C}(\bm{q},\dot{\bm{q}})\in\mathbb{R}^{N\times N}$ the generalized Coriolis matrix such that $\dot{\bm{M}}-2\bm{C}$ is skew-symmetric, and $\bm{g}(\bm{q})\coloneqq\tfrac{\partial}{\partial\bm{q}}V$ the vector of generalized potential forces derived from the potential energy $V(\bm{q})\in\mathbb{R}$.

\subsection{Problem Statement}

Having specified our considered class of EL systems, we now describe the problem setting.

\begin{objective}
The overarching goal of this article is to design a learning-based tracking controller in the form
\begin{equation}
\bm{\tau}=\bm{u}(\bm{q}_d,\dot{\bm{q}}_d,\ddot{\bm{q}}_d,\bm{e},\dot{\bm{e}},\ddot{\bm{e}}) - \bm{K}_P(\bm{\Sigma}_{\bm{\tau}})\bm{e} - \bm{K}_D(\bm{\Sigma}_{\bm{\tau}})\dot{\bm{e}} \label{gen_ctlr_obj1}
\end{equation} %, for efficiency and robustness
which exponentially stabilizes the dynamics of the error $\bm{e}=\bm{q}-\bm{q}_d$ w.r.t.~the desired trajectory $\bm{q}_d(t)\in\mathbb{R}^N$. At the same time, the nonlinear impedance $\bm{g}(\bm{q})\!+\!\bm{D}(\dot{\bm{q}})\dot{\bm{q}}$ of the system \eqref{M_C_g_system} is to be preserved by \eqref{gen_ctlr_obj1} such that the closed-loop's impedance is comprised of $\bm{g}(\bm{e})\!+\!\bm{D}(\dot{\bm{e}})\dot{\bm{e}}$ in addition to uncertainty-adaptive injections. These confidence-dependent mappings $\bm{K}_{P,D}\colon\mathbb{R}^{N\times N}\to\mathbb{R}^{N\times N}$ are designed to balance pure feedback with the model-based, mixed feedforward-feedback law $\bm{u}(\cdot)$ based on the covariance matrix $\bm{\Sigma}_{\bm{\tau}}\succ\bm{0}$.
\end{objective}

At the same time, we consider the following modeling scenario, introduced informally for the sake of space. The interested reader is referred to \cite{mein_cdc_paper} for the accurate version. 

\begin{objective}%
The unknown Lagrangian function $L(\bm{q},\dot{\bm{q}})$ from \eqref{lagrange_fund_eqs_of_mot} is approximated by a data-driven estimate $\hat{L}(\bm{q},\dot{\bm{q}})$ which is physically consistent \cite{my_tro_article} and used in the model-based controller $\bm{u}(\cdot)$. For this, we assume access to noisefree position $\bm{q}_i$ and velocity observations $\dot{\bm{q}}_i$, $i \!=\! 1,\dots,D$ with $D\in\mathbb{N}$, and to potentially noisy acceleration $\ddot{\bm{q}}_i\!+\!\bm{\alpha}_i$, $\bm{\alpha}_i \!\sim\! \mathcal{N}(\bm{0},\bm{\Sigma}_{\bm{\alpha}_i})$, and torque measurements
\begin{equation}
\bm{y}_i = \bm{\tau}_i(\bm{q}_i,\dot{\bm{q}}_i,\ddot{\bm{q}}_i) + \bm{\theta}_i \, , \quad \bm{\theta}_i \sim \mathcal{N}(\bm{0},\bm{\Sigma}_{\bm{\theta}_i}) \label{noise_meas_model_tau}\, .
\end{equation}
The noise processes $\{\bm{\alpha}_i\}$ and $\{\bm{\theta}_i\}$ are white, zero-mean, uncorrelated, and have known covariance matrices $\bm{\Sigma}_{\bm{\alpha}_i}$ and $\bm{\Sigma}_{\bm{\theta}_i}$, respectively.
After collecting all $D$ observations at input positions $\bm{X}\!=\![\bm{q}_i^\top\!,\dot{\bm{q}_i}^\top\!,\ddot{\bm{q}}_i^\top\!+\!\bm{\alpha}_i^\top]$ with analogous output matrix $\bm{Y}\!=\![\bm{y}_i^\top]$, we obtain the training data set $\mathcal{D}=\{\bm{X},\bm{Y}\}$.
\end{objective}

\section{GP Regression for Lagrangian Systems}\label{gpr_sec}

This section briefly overviews our employed modeling framework from \cite{mein_cdc_paper,my_tro_article}. For a complete introduction, particularly to GPs, the reader is referred to the literature \cite{rasmussen_gp_mit_book,raissi_gp_diff_eqs,nips_deriv_obs_solak,md_gps_alvarez}.

\subsection{Background: Gaussian Process (GP) Framework}

A Gaussian Process (GP) can be interpreted as a Gaussian distribution extended from random variables to functions. By construction, it thus inherits the properties of the Normal distribution, such that, e.g., conditioning and marginalization remain Gaussian. Considering vector-valued functions $\bm{f}\colon\mathbb{R}^M\to\mathbb{R}^N$, a GP with mean $\bm{m}(\bm{x})$ and covariance or kernel $\bm{K}(\bm{x},\bm{x}^\prime)$ is denoted by $\bm{f}(\bm{x}) \sim \mathcal{GP}\left(\bm{m}(\bm{x}), \bm{K}(\bm{x},\bm{x}^\prime) \right)$. Given $D$ observations $\bm{y}_i = \bm{f}(\bm{x}_i) + \bm{\epsilon}_i$, $i=1,\dots,D$, perturbed by white noise $\bm{\epsilon}_i\sim\mathcal{N}(\bm{0},\bm{\Sigma}_{\bm{\epsilon}_i})$, GPs assume a prior distribution of the function $\bm{f}$, specified by $\bm{m}$ and $\bm{K}$, and then leverage the resulting predictive distribution by using Bayes' rule. Equivalently, from the function-space view,  GPs exploit the joint Gaussian distribution of the measurements $\bm{Y}$ and a desired estimate $\bm{f}(\bm{x})$ by conditioning on the former, giving rise to the posterior mean $\bm{\mu}_{\bm{f}}(\bm{x})\equiv\mathrm{E}[\bm{f}(\bm{x})|\bm{Y},\bm{X}]$ and covariance $\bm{\Sigma}_{\bm{f}}(\bm{x})  \equiv\mathrm{Var}[\bm{f}(\bm{x})|\bm{Y},\bm{X}]$ given in turn by
\begin{align}
\bm{\mu}_{\bm{f}}(\bm{x}) \!\!=& \bm{m}(\bm{x}) \!+\! \bm{K}(\bm{x},\bm{X})\big(\!\bm{K}(\!\bm{X}\!,\!\bm{X}\!)\!\!+\!\!\bm{\Sigma}_{\bm{\epsilon}}\!\big)^{-1} \!\mathrm{vec}(\bm{Y}\!\!\!-\!\bm{m}(\bm{X})),\nonumber\\ %\Delta\bm{y} \\
\bm{\Sigma}_{\bm{f}}(\bm{x}) \!\!=& \bm{K}(\bm{x},\bm{x}) \!-\! \bm{K}(\bm{x},\bm{X}) \big(\bm{K}(\bm{X},\bm{X})\!+\!\bm{\Sigma}_{\bm{\epsilon}}\big)^{-1} \! \bm{K}(\bm{X},\bm{x}) ,\nonumber
\end{align}
where $\bm{K}(\bm{X},\bm{X})$ and $\bm{\Sigma}_{\bm{\epsilon}}$ are the multidimensional Gramian and noise covariance block matrices, respectively.

The kernel matrix $\bm{K}(\bm{x},\bm{x}^\prime)\in\mathbb{R}^{N\times N}$ is positive semidefinite \cite{md_gps_alvarez} for any $\bm{x},\bm{x}^\prime$, quantifies the correlation or similarity between the components of $\bm{f}(\bm{x})$ and $\bm{f}(\bm{x}^\prime)$, and determines higher-level functional properties such as smoothness. Dependent on so-called hyperparameters, the marginal likelihood \cite{rasmussen_gp_mit_book} is mostly maximized numerically to maximize the probability of observing the measured outputs at the given inputs.

Also, GPs have the key property that linear transformations remain GPs \cite{raissi_gp_diff_eqs} due to their construction based on the expectation. Thus, applying a linear transformation operator $\mathcal{T}_{\bm{x}}$, e.g., differentiation or integration, yields a GP again \cite{rath_gp_ldeqs}
\begin{equation} % Linear Operator GP
\mathcal{T}_{\bm{x}}\bm{f}(\bm{x}) \sim \mathcal{GP}\left(\mathcal{T}_{\bm{x}}\bm{m}(\bm{x}), \mathcal{T}_{\bm{x}} \bm{K}(\bm{x},\bm{x}^\prime) \mathcal{T}^\top_{\bm{x}^\prime} \right) \, . \label{GP_linear_op}
\end{equation}

Having introduced the basic framework of GPs, we now describe their unification with Lagrangian first-order principles.% \rev{Based on \cite{mein_cdc_paper}, we extend the method's applicability from conservative to dissipative systems, thus making up our first contribution.}

\subsection{Background: Lagrangian-Gaussian Process (L-GP)}

\subsubsection{Modeling Approach}

The core concept we apply is to physically constrain the employed GP distribution's function space. For this, we exploit the Lagrangian-differential operator
\begin{equation} % Lagrangian operator
\mathcal{L}_{\bm{q}} \coloneqq \left( \tfrac{\partial}{\partial\dot{\bm{q}}\!^\top}\ddot{\bm{q}} + \tfrac{\partial}{\partial\bm{q}\!^\top} \dot{\bm{q}} \right)\tfrac{\partial}{\partial\dot{\bm{q}}} - \tfrac{\partial}{\partial\bm{q}} \label{lagrange_op}
\end{equation}
inducing the multidimensional dynamics \eqref{M_C_g_system} from a scalar GP for the uncertain Lagrangian function $L(\bm{q},\dot{\bm{q}})$ with mean $m_L$ and kernel $k_L$. Thus, Hamilton's principle of Least Action is deterministically guaranteed by embedding the differential equation structure~\eqref{lagrange_fund_eqs_of_mot} into the GP. Since the transformation \eqref{lagrange_op} is linear, we obtain the multidimensional GP
\begin{equation} % Multidimensional tau vector GP
\bm{\tau}_{c}(\bm{q},\dot{\bm{q}},\ddot{\bm{q}}) \sim \mathcal{GP}\left(\mathcal{L}_{\bm{q}} m_L(\bm{q},\dot{\bm{q}}), \mathcal{L}_{\bm{q}} \mathcal{L}_{\bm{q}^\prime}^\top k_L(\bm{q},\dot{\bm{q}},\bm{q}^\prime,\dot{\bm{q}}^\prime)\right) \label{tau_vec_l_gp}
\end{equation}%
for the conservative torques.%, where $\bm{x}=(\bm{q},\dot{\bm{q}},\ddot{\bm{q}})$.

Incorporating generalized friction in compliance with Lagrangian mechanics into structured model learning is nontrivial \cite{lutter2021combining}. However, the specific matrix-vector structure asserted in Assumption~\ref{ass_disp} can be enforced with the covariance \cite{my_tro_article} 
\begin{equation}
\bm{K}_f(\dot{\bm{q}},\dot{\bm{q}}^\prime) = \mathrm{diag}(\dot{\bm{q}})\bm{K}_d(\dot{\bm{q}},\dot{\bm{q}}^\prime)\mathrm{diag}(\dot{\bm{q}}^\prime)  \label{dsp_cov_expr}
\end{equation}
for which positive-semidefiniteness and passivity guarantees can be provided \cite{rui_l4dc}. Therefore, combining the dissipative GP $\bm{\tau}_f\sim\mathcal{GP}(\bm{m}_f,\bm{K}_f)$ for prior mean $\bm{m}_f$ and kernel \eqref{dsp_cov_expr} with the conservative torques \eqref{tau_vec_l_gp}, we arrive at the multidimensional composite GP model
\begin{equation}
\bm{\tau}=\bm{\tau}_c+\bm{\tau}_f \sim \mathcal{GP}\left(\mathcal{L}_{\bm{q}} m_L + \bm{m}_f, \mathcal{L}_{\bm{q}} \mathcal{L}_{\bm{q}^\prime}^\top k_L + \bm{K}_f\right) \label{dsp_lgp} \, .
\end{equation}

\subsubsection{Energy Structuring}

The Lagrangian-GP in \eqref{lagrange_op}--\eqref{tau_vec_l_gp} is further split up into its energy components $L=T-G-U$ for each of which we consider an underlying independent GP again, i.e., $T,G,U\sim\mathcal{GP}$. In particular, for the kinetic energy, we make use of a specific kernel structure \cite{mein_cdc_paper} given by
\begin{equation} % Quadratic functional
k_T = \frac{1}{4} \dot{\bm{q}}^\top \mathrm{diag}(\dot{\bm{q}}^{\prime}) \bm{\Theta}_{\bm{M}}(\bm{q},\bm{q}^\prime) \mathrm{diag}(\dot{\bm{q}}^{\prime}) \dot{\bm{q}} \label{quadratic_kernel_functional}
\end{equation}
with the Cholesky decomposed covariance $\bm{\Theta}_{\bm{M}} = \bm{R}_{\bm{M}}^\top\bm{R}_{\bm{M}}$ and upper-right triangular $\bm{R}_{\bm{M}}(\bm{q},\bm{q}^\prime)$. The elastic energy GP $U(\bm{q})$ is constrained analogously. Contrary to deterministically enforcing positive definiteness as in \cite{lutter2019deep} via Cholesky decomposition of the output matrix, e.g., the mass-inertia, this approach only deterministically preserves the energies' quadratic form in order to stochastically preserve Gaussianity. In this manner, the mass-inertia matrix remains a (symmetric) matrix-valued GP.

Due to the structurally enforced constraints \eqref{lagrange_op}--\eqref{quadratic_kernel_functional}, the posterior $\hat{\bm{\tau}}\!\!\coloneqq\!\!\bm{\mu}_{\bm{\tau}}\!\!=\!\!\mathrm{E}[\bm{\tau}|\bm{Y}\!\!,\!\bm{X}]\!$ of the L-GP \eqref{dsp_lgp} can be written as 
\begin{equation}
\hat{\bm{\tau}}(\bm{q},\dot{\bm{q}},\ddot{\bm{q}}) = \hat{\bm{M}}(\bm{q})\ddot{\bm{q}} + \hat{\bm{C}}(\bm{q},\dot{\bm{q}})\dot{\bm{q}} + \hat{\bm{g}}(\bm{q}) + \hat{\bm{D}}(\dot{\bm{q}})\dot{\bm{q}} \label{l_gp_post_mu}\, .
\end{equation}
The (symmetric) posterior mass-inertia estimate $\hat{\bm{M}}\colon\mathbb{R}^N\to\mathbb{R}^{N\times N}$ is guaranteed to be positive definite with high probability~\cite{mein_cdc_paper}. Also, $\hat{\bm{C}}$ is here the Coriolis estimate constructed as $\hat{\bm{C}}(\bm{q},\dot{\bm{q}})\!\coloneqq\!\frac{1}{2}(\!\tfrac{\partial^2\hat{T} }{\partial\dot{\bm{q}}\partial\bm{q}\!^\top} \!+\! \dot{\hat{\bm{M}}}(\bm{q}) \!-\! \tfrac{\partial^2\hat{T} }{\partial\bm{q}\partial\dot{\bm{q}}\!^\top}\!)$~\cite{murray_amitrm}. These first two terms, combined with the potential force estimate $\hat{\bm{g}}(\bm{q})\!=\!\tfrac{\partial}{\partial\bm{q}}(\hat{G}\!+\!\hat{U})$, also comprise the posterior $\hat{\bm{\tau}}_c\!\coloneqq\!\bm{\mu}_{\bm{\tau}_c}$ of the conservative torque GP \eqref{tau_vec_l_gp}, and are provably lossless \cite{mein_cdc_paper}, deterministically.%, due to the differential structural embedding \eqref{lagrange_nabla_eqs_of_mot} and the skew-symmetry of $\dot{\hat{\bm{M}}}-2\hat{\bm{C}}$. 

Having laid out the modeling framework, we can now proceed with the first contributions of this article.

\section{Control Design}\label{sec_ctrl_dsgn}

In the following, we present our proposed feedback control architecture targeting robust and accurate yet efficient trajectory tracking. 

\subsection{Natural Structure-Preserving Control}

We propose a structure-preserving control scheme based on projections according to
\begin{align}
\bm{\tau} \!=\! \hat{\bm{M}}(\bm{q})\ddot{\bm{q}}_d \!+\! \hat{\bm{C}}(\bm{q},\dot{\bm{q}})\dot{\bm{q}}_d \!&+\! \left(\bm{I}\!\!-\!\!h(\bm{e}\!^\top\!\hat{\bm{g}}_{\bm{q}})\bm{P}_{\bm{e}}\right)\!\hat{\bm{g}}_{\bm{q}} \!-\! \bm{g}_d(\bm{e})\nonumber\\ &+\!\big(\bm{I}\!\!-\!\! h(\dot{\bm{e}}\!^\top\!\hat{\bm{d}}_{\dot{\bm{q}}})\bm{P}_{\dot{\bm{e}}}\big)\hat{\bm{d}}_{\dot{\bm{q}}} \!-\! \bm{d}_d(\dot{\bm{e}}) \label{c_law}
\end{align}
where $\bm{e}=\bm{q}-\bm{q}_d$ is the position error based on the desired reference trajectory $\bm{q}_d(t)$, and $\hat{\bm{M}}(\bm{q})$, $\hat{\bm{C}}(\bm{q},\dot{\bm{q}})$, $\hat{\bm{g}}(\bm{q})\eqqcolon\hat{\bm{g}}_{\bm{q}}$ and $\hat{\bm{D}}(\dot{\bm{q}})\dot{\bm{q}}\eqqcolon\hat{\bm{d}}_{\dot{\bm{q}}}$ are the learned inertia, Coriolis matrix, gravity and friction, respectively, as defined in \eqref{l_gp_post_mu}. The Heaviside step function $h(x)$ is given by % $w_{\bm{e}}(\bm{g}),w_{\dot{\bm{e}}}(\bm{d})$ defined equivalently according to
\begin{equation}
h(x) = \tfrac{1}{2}\left(1+\mathrm{sign}(x)\right) \label{heaviside_func}
\end{equation}
for $x\in\mathbb{R}$, and $\bm{P}_{\bm{e}}\in\mathbb{R}^{N\times N}$ is the orthogonal projector \cite{roy_la_ma4s} onto the one-dimensional space spanned by $\bm{e}\in\mathbb{R}^N$% defined as
\begin{equation}
\bm{P}_{\bm{e}}=\tfrac{\bm{e}\bm{e}^\top}{\|\bm{e}\|^2}	\, . \label{proj_mat}
\end{equation}
The methodology behind \eqref{c_law} is to alternate the well-known standard PD+ controller~\cite{pdp_ref} to compensate potential and dissipative forces stemming from, e.g., gravity and friction, respectively, only when necessary. Thus, if either of the model-based estimates points in the direction of the errors, the parallel components are preserved and exploited, while orthogonal subvectors are canceled. In this way, we can improve the actuation efficiency of \cite{pdp_ref}. Note that this is also implicitly achieved in part by combining the PD+ controller~\cite{pdp_ref} with the increased modeling accuracy of the L-GP framework~\cite{mein_cdc_paper}. Although, in total, four pairings of the impedance $\bm{g}(\bm{q})\!+\!\bm{D}(\bm{q},\dot{\bm{q}})\dot{\bm{q}}$ with positional and velocity errors would be possible, our usage of the two in \eqref{c_law} leads to variable yet natural PD gains exclusively depending on the error variable they multiply, a typical design approach for adaptive feedback gains~\cite{ravichandran}. Also, note that the controller \eqref{c_law} is continuously differentiable w.r.t.~the arguments of the Heaviside functions since these are zero only when the respective vectors become orthogonal. Due to the multiplication with the respective projections, the parallel subvectors also become zero in these instances, enforcing a smooth mapping.

In the following subsection, we elaborate on the role of the remaining quantities in the controller \eqref{c_law}, i.e., the desired potential and dissipation forces $\bm{g}_d(\bm{e})$ and $\bm{d}_d(\dot{\bm{e}})$, respectively.

\subsection{Uncertainty-based Potential Shaping and Damping Injection: Feedforward-Feedback Balancing}

In order to obtain a structurally preserved closed-loop, we propose a reference potential and dissipation given by%In order to naturally generate a control action while preserving the system's structure; we propose a reference potential and dissipation given by
\begin{subequations}\label{des_dmp}
\begin{align}
\bm{g}_d(\bm{e},t) &= \hat{\bm{g}}(\bm{e}) + \bm{K}_P(\bm{\Sigma_{\tau}}(t))\bm{e} \label{des_pot_frcs} \\
\bm{d}_d(\dot{\bm{e}},t) &= \hat{\bm{d}}(\dot{\bm{e}}) + \bm{K}_D(\bm{\Sigma_{\tau}}(t))\dot{\bm{e}} 
\end{align}
\end{subequations}
with (time-variant) uncertainty-adaptive energy shaping and damping injection based on the covariance matrix $\bm{\Sigma_{\tau}}$ and positive definite matrix functions $\bm{K}_{P,D}\colon \mathbb{R}^{N\times N} \to \mathbb{R}^{N\times N}$. Aside from achieving robustness by this design approach due to the minimal dynamical intervention, this is also particularly beneficial upon potential interactions with the environment, such as for a soft robot remaining inherently compliant due to its preserved softness \cite{cosimo_ijrr_plan_soro}, instead of being stiffened by feedback. To still be able to accurately achieve the control task, the uncertainty-adaptive gains $\bm{K}_{P,D}(\bm{\Sigma_{\tau}})$ have the purpose of implementing an automatic balancing between feedforward and feedback elements. 

\subsubsection{Closed-Loop Dynamics}

To begin with, we assume zero modeling errors $\bm{g}\!\equiv\!\hat{\bm{g}},\bm{d}\!\equiv\!\hat{\bm{d}}$ with $\bm{K}_P(\bm{0})\!=\!\bm{K}_D(\bm{0})\!=\!\bm{0}$ and linear matrix-vector mappings $\bm{g}\!=\!\bm{K}\bm{q},\bm{d}\!=\!\bm{D}\dot{\bm{q}}$ for constant $\bm{0}\!\prec\!\bm{K},\bm{D}\!\in\!\mathbb{R}^{N\times N}$, e.g., for linear gravity-compensated spring-damper systems, this choice effectively reduces \eqref{c_law} to a feed-forward controller:
\begin{equation}
\bm{\tau} = \hat{\bm{M}}\ddot{\bm{q}}_d \!+\! \hat{\bm{C}}\dot{\bm{q}}_d \!+\! \bm{g}(\bm{q}_d) \!+\! \bm{d}(\dot{\bm{q}}_d) \!-\! h(\bm{e}\!^\top\!\bm{g})\bm{g}_{\parallel\bm{e}} \!-\! h(\dot{\bm{e}}\!^\top\!\bm{d})\bm{d}_{\parallel\dot{\bm{e}}} \notag
\end{equation}
with $\bm{g}_{\parallel\bm{e}}\!\!\coloneqq\!\!\bm{P_e}\bm{g}$, $\bm{d}_{\parallel\dot{\bm{e}}}\!\!\coloneqq\!\!\bm{P}_{\dot{\bm{e}}}\bm{d}$. Here, feedback is only present if the according error vectors point in the direction of the respective gravity and friction vectors, thus leading to a structural amplification, i.e., doubling, of the parallel subvectors $\bm{g}_{\parallel \bm{e}},\bm{d}_{\parallel \dot{\bm{e}}}$.

Next, considering in a more general scenario the model error $\tilde{\bm{\tau}}\!=\!\bm{\tau}\!-\!\hat{\bm{\tau}}$ and nonlinear vector maps $\bm{g},\bm{d}$, we follow that the controller \eqref{c_law} results in the closed-loop system
\begin{align}
\hat{\bm{M}}(\bm{q})\ddot{\bm{e}} + \hat{\bm{C}}(\bm{q},\dot{\bm{q}})\dot{\bm{e}} &+ h(\dot{\bm{e}}^\top\hat{\bm{d}}(\dot{\bm{q}}))\tfrac{\dot{\bm{e}}^\top\hat{\bm{d}}(\dot{\bm{q}})}{\|\dot{\bm{e}}\|^2}\dot{\bm{e}} + \bm{d}_d(\dot{\bm{e}}) \label{error_dyn}\\
&+ h(\bm{e}^\top\hat{\bm{g}}(\bm{q}))\tfrac{\bm{e}^\top\hat{\bm{g}}(\bm{q})}{\|\bm{e}\|^2}\bm{e} + \bm{g}_d(\bm{e}) + \tilde{\bm{\tau}} = \bm{0} \nonumber \, .
\end{align}
It becomes clear that the control law \eqref{c_law} leads to error dynamics with adaptively increasing damping and stiffness coefficients stemming from the structural preservation of the system's dynamics. Note that the dynamics \eqref{error_dyn} are also continuously differentiable and smooth w.r.t.~the Heaviside functions' arguments. Moreover, despite the normalizations w.r.t.~the squared lengths of the error variables as part of the projectors in \eqref{c_law}, the dynamics \eqref{error_dyn} remain bounded, since the projection matrices \eqref{proj_mat} only give the parallel component $\hat{\bm{g}}_{\parallel\bm{e}}\!\!\coloneqq\!\!\frac{\bm{e}^\top\!\hat{\bm{g}}}{\|\bm{e}\|^2}\bm{e}\!\!=\!\!\bm{P_e}\hat{\bm{g}}$, i.e., $\lVert\bm{P}_{\bm{e}}\rVert\!\!=\!\!1$ such that $\lVert \bm{P_e}\hat{\bm{g}}\rVert \!\leq\! \lVert\hat{\bm{g}}\rVert$. %the projection is bounded since

\subsubsection{Variance-based Gain Adaptation}

For uncertainty-adaptive balancing between the feedback and feedforward elements in \eqref{c_law}--\eqref{des_dmp}, we propose a specific structure for the matrix gains $\bm{K}_{P,D}$ given by
\begin{equation}
\bm{K}(\bm{\Sigma_{\tau}}) \coloneqq \bm{K}_1\!\left(\!\bm{I}-\left[\bm{K}_3(\bm{K}_2\!+\!\bm{\Sigma_{\tau}})\bm{K}_3+\bm{K}_1\right]^{-1}\!\bm{K}_1\!\right) \label{sigma_gain_adapt_law}
\end{equation}
with constant, positive definite matrices $\bm{K}_{1-3}\succ\bm{0}$. Using variational principles for the eigenvalues of symmetric operators \cite{bhatia_ma}, we can provide a guarantee on the user-definably bounded interval in which the eigenvalues of $\bm{K}(\bm{\Sigma_{\tau}})$ lie.

% Lemma 1
\begin{lemma}\label{lemma_one} The uncertainty-adaptive gain \eqref{sigma_gain_adapt_law} with constant $\bm{0} \prec \underline{k}_i\bm{I} \prec \bm{K}_i \prec \bar{k}_i\bm{I}$, where $\underline{k}_i,\bar{k}_i\in\mathbb{R}^+$ for $i=1,2,3$, is guaranteed to fulfill the linear matrix inequalities 
\begin{subequations}
\begin{align}
\tfrac{1}{(\underline{k}_3^2\underline{k}_2)^{-1}+\underline{k}_1^{-1}}\bm{I} &\prec \bm{K}(\bm{\Sigma_{\tau}}) \prec \bar{k}_1\bm{I} \, , \label{bound_k_gain}\\
-\left(\tfrac{\bar{k}_1\bar{k}_3}{\underline{k}_3^2\underline{k}_2+\underline{k}_1}\right)^2 |\underline{\dot{\sigma}}_{\bm{\tau}}|\bm{I} &\prec \dot{\bm{K}}(\bm{\Sigma_{\tau}}) \prec \left(\tfrac{\bar{k}_1\bar{k}_3}{\underline{k}_3^2\underline{k}_2+\underline{k}_1}\right)^2 \bar{\dot{\sigma}}_{\bm{\tau}}\bm{I} \, , \label{bound_dk_dt}
\end{align}%,\!\dot{\bm{\Sigma}}_{\bm{\tau}}
\end{subequations}
where the second inequality pair assumes $-|\underline{\dot{\sigma}}_{\bm{\tau}}|\bm{I} \prec \dot{\bm{\Sigma}}_{\bm{\tau}} \prec \bar{\dot{\sigma}}_{\bm{\tau}}\bm{I}$ holds for the covariance derivative $\dot{\bm{\Sigma}}_{\bm{\tau}}$ with $\underline{\dot{\sigma}}_{\bm{\tau}}\leq 0$.
\end{lemma}

\begin{proof}
For the proof of the bounds for $\bm{K}(\bm{\Sigma_{\tau}})$, we reformulate \eqref{sigma_gain_adapt_law} to $\bm{K}(\bm{\Sigma_{\tau}})=(\bm{K}_1^{-1} + \bm{K}_3^{-1}(\bm{K}_2+\bm{\Sigma_{\tau}})^{-1}\bm{K}_3^{-1})^{-1}$ by making use of the matrix inversion lemma \cite[C.4.3]{boyd_convex_opt}. The inequalities~\eqref{bound_k_gain} then directly follow from \cite[Lemma~3]{my_tro_article} after utilizing $\bm{\lambda}(\bm{A}^{-1})=\bm{\lambda}^{-1}(\bm{A})$, $\bm{A}\in\mathbb{R}^{N\times N}$. For the time derivative $\dot{\bm{K}}(\bm{\Sigma_{\tau}})$, we exploit $\tfrac{d}{dt}\bm{A}^{-1}=-\bm{A}^{-1}\dot{\bm{A}}\bm{A}^{-1}$,  and compute $\dot{\bm{K}}(\bm{\Sigma}_{\bm{\tau}})=\bm{K}_1 \tilde{\bm{K}}^{-1}(\bm{\Sigma}_{\bm{\tau}}) \bm{K}_3 \dot{\bm{\Sigma}}_{\bm{\tau}} \bm{K}_3 \tilde{\bm{K}}^{-1}(\bm{\Sigma}_{\bm{\tau}}) \bm{K}_1$, where $\tilde{\bm{K}}(\bm{\Sigma}_{\bm{\tau}})=\bm{K}_3(\bm{K}_2+\bm{\Sigma_{\tau}})\bm{K}_3+\bm{K}_1$. Finally, applying Lidskii's corollary \cite[III.4.6]{bhatia_ma} for the eigenvalues of the product of positive definite matrices, we arrive at \eqref{bound_dk_dt}.% finish the proof.
\end{proof}

Using \eqref{bound_k_gain}, one can enforce the spectrum $\bm{\lambda}(\bm{K}(\bm{\Sigma_{\tau}}))$ of the adaptive gain \eqref{sigma_gain_adapt_law} to lie in an interval on the axis of positive reals, where $\bar{k}_1$ defines the maximum gain for $\bar{\lambda}(\bm{\Sigma_{\tau}})\to\infty$. We gather that $\bm{K}_3$ defines the scaling manner in which the gain follows $\bm{\Sigma_{\tau}}$. Moreover, $\bm{K}_2$ is typically chosen such that $0<\underline{k}_2<<1$ to specify a small positive minimum gain while also improving the condition number of the covariance. A desired positive lower bound $0<\underline{k}<\bar{k}=\underline{k}_1$ can be achieved by, e.g., fixing $\underline{k}_1,\underline{k}_2$ and setting $1/\underline{k}_3=\sqrt{\underline{k}_2(1/\underline{k}-1/\underline{k}_1)}$.%, while $\underline{k}=0$ can be set via $\bm{K}_2=0$.

Having motivated and legitimated the design principles behind our proposed control law \eqref{c_law}--\eqref{des_dmp}, we now proceed to analyze the stability and convergence of its resulting closed-loop error dynamics \eqref{error_dyn} in the following subsection.

\subsection{Exponential Stability Analysis}

Due to the data-driven nature of the model, we consider its uncertainty probabilistically and arrive at a stability guarantee holding with high probability $1-\delta$, where $0<\delta\ll 1$. 
We now formulate the main theoretical result certifying exponential stability and convergence of the proposed control law \eqref{c_law}.% . %  law

\begin{theorem}[Natural $\bm{\Sigma}$-adaptive PD+]\label{lgp_nat_pdp_theo1}
The L-GP-based closed-loop~\eqref{error_dyn} is exponentially stable to within the ball %stable with probability
\begin{equation}
B(\varrho) \equiv \Big\{ \lVert[\bm{e},\dot{\bm{e}}](t)\rVert \leq \varrho(t) = \Delta\sqrt{\tfrac{\varepsilon/\vartheta + 1/\varphi}{2\underline{\mu}(t)}} \Big\} , \label{th1_ball}
\end{equation}
where the convergence rate $\alpha(t)\geq\underline{\alpha}\in\mathbb{R}^+$ is given by
\begin{equation}
\alpha(t) = \tfrac{ \lambda(\alpha(t),t)(\|\bm{e}\|^2 \!+\! \|\dot{\bm{e}}\|^2) + \varepsilon(\bm{e}^\top\!\hat{\bm{g}}_{\bm{e}} + \nu_{\bm{e}} + \omega_{\dot{\bm{e}}}) + \nu_{\dot{\bm{e}}} + \omega_{\bm{e}} }{\hat{G}(\bm{e})} \label{implct_eq_alpha}
\end{equation}
with the L-GP's potential energy $\hat{G}(\bm{e})$, an eigenvalue function $\lambda(\alpha(t),t)\!\in\!\mathbb{R}$, and $\nu_{\bm{e}} \!\!=\!\! h(\bm{e}^\top\!\hat{\bm{g}}_{\bm{q}})\bm{e}^\top\!\hat{\bm{g}}_{\bm{q}} \!\!\geq\!\! 0$, $\nu_{\dot{\bm{e}}} \!\!=\!\! h(\dot{\bm{e}}^\top\!\hat{\bm{d}}_{\dot{\bm{q}}})\dot{\bm{e}}^\top\!\hat{\bm{d}}_{\dot{\bm{q}}} \!\!\geq\!\! 0$ and $\omega_{\bm{e}} \!\!=\!\! h(\bm{e}^\top\!\hat{\bm{g}}_{\bm{q}})\dot{\bm{e}}^\top\!\bm{P}_{\bm{e}}\hat{\bm{g}}_{\bm{q}}$,  $\omega_{\dot{\bm{e}}} \!\!=\!\! h(\dot{\bm{e}}^\top\!\hat{\bm{d}}_{\dot{\bm{q}}})\bm{e}^\top\!\bm{P}_{\dot{\bm{e}}}\hat{\bm{d}}_{\dot{\bm{q}}}$ for $\hat{\bm{g}}_{\bm{q}}\coloneqq\hat{\bm{g}}({\bm{q}})$, $\hat{\bm{d}}_{\dot{\bm{q}}}\coloneqq\hat{\bm{d}}(\dot{\bm{q}})$ and the heaviside function $h(\cdot)$ from \eqref{heaviside_func}. The statement holds with an exact probability of
\begin{equation}
\mathrm{Pr}\left\{\lVert[\bm{e},\dot{\bm{e}}](t)\rVert \leq \varrho(t) + c_0e^{-\int_{t_0}^{t}\alpha(\tau)d\tau}\right\} = 1 - \delta \label{prob_theo1}
\end{equation}
for $\delta\in(0,1)$. Aside from the constants $\varepsilon,\vartheta,\Delta\in\mathbb{R}^+$, the radius $\varrho(t)$ is determined by the coordinate metric eigenvalue
\begin{equation}
\underline{\mu}(t) \!=\! \tfrac{\kappa + \hat{\underline{m}}(t)}{2} - \sqrt{[\tfrac{\kappa-\hat{\underline{m}}(t)}{2}]^2 \!+\! [\varepsilon\hat{\underline{m}}(t)]^2} + \tfrac{2\hat{G}(\bm{e})}{\|\bm{e}\|^2\!+\!\|\dot{\bm{e}}\|^2} \,. \label{mu_lb_th1}
\end{equation}
The virtual stiffness $\kappa\in\mathbb{R}^+$ and scale $\varphi\in\mathbb{R}^+$ are given by
\begin{subequations}
\begin{align}
\kappa &= \underline{k}_P + \varepsilon(\underline{d} - \underline{\alpha}\hat{m}_{\Sigma})  \label{opt_kappa} \\
\varphi &= 2[\underline{d} - \varepsilon(\underline{k}_P\!-\!\tfrac{\vartheta}{2})\!+\! \underline{\alpha}\kappa] - (\varepsilon\!+\!\underline{\alpha})\hat{m}_{\Sigma} \, , \label{opt_phi}  
\end{align}
\end{subequations}
where $\underline{d} \!=\! \underline{\hat{d}} \!+\! \underline{k}_D$, $\hat{m}_{\Sigma}\!=\!\underline{\hat{m}}\!+\!\bar{\hat{m}}$ are constant worst-case bounds, i.e., $\underline{d}\bm{I}\!\preceq\!\hat{\bm{D}}\!+\!\bm{K}_D$, $\underline{\hat{m}}\bm{I}\!\preceq\!\hat{\bm{M}}\!\preceq\!\bar{\hat{m}}\bm{I}$. The inequalities 
\begin{subequations}
\begin{align}
\varepsilon &< \min\!\left[\! \tfrac{\underline{d}}{\underline{k}_P-\vartheta/2+\hat{m}_{\Sigma}} , \, \tfrac{\underline{d} - \underline{\alpha}\hat{m}_{\Sigma}}{2\bar{\hat{m}}} \!\left(\! 1 \!+\! \sqrt{1\!\!+\!\!\tfrac{4\bar{\hat{m}}^2\underline{k}_P}{(\underline{d} - \underline{\alpha}\hat{m}_{\Sigma})^2}} \right) \! \right] \label{eps_intrvl_theo1}\\
\vartheta &< 2(\underline{k}_P+\hat{m}_{\Sigma}) \\
\underline{\alpha} &< \min\!\left[ \!\tfrac{\underline{d}}{\hat{m}_{\Sigma}} , \, \alpha_0 \!+\! \sqrt{ \alpha_0^2 \!+\! \tfrac{\underline{d}-\varepsilon(\underline{k}_P-\vartheta/2+\hat{m}_{\Sigma})}{\varepsilon\hat{m}_{\Sigma}} }  \right] \, .
\end{align}
\end{subequations}
are fulfilled by the parameters $\varepsilon,\vartheta$ and $\underline{\alpha}$ for validity of the statement, where $\alpha_0\!\coloneqq\!\tfrac{\underline{k}_P \!+\! \varepsilon \underline{d} \!-\! \hat{m}_{\Sigma}/2}{2\varepsilon\hat{m}_{\Sigma}}$.
\end{theorem}%$\lVert[\bm{e}_0,\dot{\bm{e}}_0]\rVert\leq\bar{e}$

\begin{remark}
Note that the implicit equation \eqref{implct_eq_alpha} can be solved explicitly for the rate $\alpha(t)$. An analytical expression, also for the eigenvalue function $\lambda\colon\mathbb{R}^+\times\mathbb{R}\to\mathbb{R}$, is derived and given in the proof of Theorem~\ref{lgp_nat_pdp_theo1}.
\end{remark}

For the proof of Theorem~\ref{lgp_nat_pdp_theo1}, we apply ideas from contraction theory \cite{lohmiller_ct}, revolving around the problem of finding a coordinate transformation, to construct an optimal Lyapunov function, which allows the provision of a maximal convergence rate and a minimal ball radius. This enables the derivation of a strong theoretical guarantee, which depends on a conservative usage of Young's inequality only once each where it is inevitable, i.e., where position and velocity are multiplied with the model error. The repeated application of these conservative inequalities, as in, e.g., \cite{thomas_tracking_control_of_el_systems}, for bounds on Lyapunov functions and Coriolis terms, can potentially lead to infeasible guarantees. To circumvent these issues, we build majorly on matrix analysis arguments, for which we provide the following sufficient result as an intermediate step.%, cf. also Appendix~\ref{classical_proof} for the classical approach applied to \eqref{error_dyn}

\begin{lemma}\label{lemma_met}
The symmetric coordinate metric
\begin{equation}
\bm{\mathcal{M}} = \begin{bmatrix} \bm{\mathcal{K}} & \varepsilon \hat{\bm{M}} \\ \varepsilon \hat{\bm{M}} & \hat{\bm{M}} \end{bmatrix} , \label{coord_metric}
\end{equation}
with a real constant $\varepsilon\in\mathbb{R}$, design stiffness $\bm{0}\prec\underline{\kappa}\bm{I}\preceq\bm{\mathcal{K}}\preceq\bar{\kappa}\bm{I}$ and L-GP inertia estimate $\bm{0}\prec\underline{\hat{m}}\bm{I}\preceq\hat{\bm{M}}\preceq\bar{\hat{m}}\bm{I}$, is uniformly positive definite, i.e., $\bm{\mathcal{M}}\succ\bm{0} \, \forall t$, iff $|\varepsilon| < \sqrt{\underline{\kappa}/\bar{\hat{m}}}$. Moreover, the metric \eqref{coord_metric} is guaranteed to fulfill the LMI bounds
\begin{subequations}\label{lmi_bounds_metric}
\begin{align}
\bm{0} &\prec \underline{\mu}\bm{I} \preceq \bm{\mathcal{M}} \preceq \bar{\mu}\bm{I} \label{lmi_metric} \\
\underline{\bar{\mu}} &= \tfrac{1}{2}\left(\underline{\bar{\kappa}}+\underline{\bar{\hat{m}}} \pm \sqrt{(\underline{\bar{\kappa}}-\underline{\bar{\hat{m}}})^2 + (2\varepsilon\bar{\hat{m}})^2}\right) \label{mu_bar_unli}
\end{align}
\end{subequations}
for nonzero values $0 < |\varepsilon| < \sqrt{\underline{\hat{m}}\underline{\kappa}} / \bar{\hat{m}}$.
\end{lemma}

\begin{proof}
The proof is given in Appendix~\ref{l3_p_appendix}.
\end{proof}

\noindent Next, we provide an exact and necessary condition in the simplified case of an Euclidian virtual stiffness.

\begin{corollary}\label{coro_met}
For Euclidian $\bm{\mathcal{K}}=\kappa\bm{I}\succ\bm{0}$ with $\kappa\in\mathbb{R}^+$ and $\hat{\bm{m}}\coloneqq\bm{\lambda}(\hat{\bm{M}})$, all eigenvalues of the metric \eqref{coord_metric} are given by 
\begin{equation}
\bm{\lambda}(\bm{\mathcal{M}}) = \tfrac{1}{2}\left(\kappa+\hat{\bm{m}} \pm \sqrt{(\kappa-\hat{\bm{m}})^2 + (2\varepsilon\hat{\bm{m}})^2}\right) \label{eig_vals_metric_eucl_K},
\end{equation}
which are positive for $|\varepsilon|<\sqrt{\kappa/\bar{\hat{m}}}$.
\end{corollary}

\begin{proof}
Please see Appendix~\ref{c1_p_appendix}.
\end{proof}

\noindent We are now ready to prove the main result.

\begin{proof_th}
To find an optimal Lyapunov function allowing the provision of a maximally tight guarantee, let us first consider the generalized candidate
\begin{equation}
V(\bm{e},\dot{\bm{e}},t) = \tfrac{1}{2} \begin{bmatrix} \bm{e}^\top \!\!&\!\! \dot{\bm{e}}^\top \end{bmatrix} \bm{\mathcal{M}}(\bm{e},\dot{\bm{e}},t) \begin{bmatrix} \bm{e} \\ \dot{\bm{e}} \end{bmatrix} \eqcolon \tfrac{1}{2} \bm{z}^\top\bm{z} \, , \label{V_ct_relatship}
\end{equation} 
representing the squared length of some adequately transformed error coordinates $\bm{z}\coloneqq\bm{\Theta}\bm{x}$. Thus, we implicitly define a positive definite and continuously differentiable metric
\begin{equation}
\bm{\mathcal{M}}(\bm{e},\dot{\bm{e}},t)=\bm{\Theta}^\top\bm{\Theta} \label{metric_lyap}
\end{equation}
in whose Riemann space the squared length of  $\bm{x}^\top\coloneq[\bm{e}^\top \, \dot{\bm{e}}^\top]$ is investigated. Contrary to the typical formulation based on virtual displacements $\delta\bm{x}$, requiring the computation of the system's Jacobian, we exploit the structure of \eqref{error_dyn} reformulating to $\dot{\bm{x}}=\bm{A}(\bm{x},t)\bm{x} + \bm{\pi}(\bm{x},t)$ with nonlinear system matrix
\begin{equation}
\bm{A}^\top(\bm{x},t) = \begin{bmatrix} \bm{0} & -\hat{\bm{M}}^{-1}(\tfrac{\nu_{\bm{e}}}{\|\bm{e}\|^2}\bm{I} \!+\! \hat{\bm{K}}_G \!+\! \bm{K}_P) \\ \bm{I} & -\hat{\bm{M}}^{-1}(\hat{\bm{C}} \!+\! \tfrac{\nu_{\dot{\bm{e}}}}{\|\dot{\bm{e}}\|^2}\bm{I} \!+\! \hat{\bm{D}} \!+\! \bm{K}_D) \end{bmatrix} \label{A_pi}
\end{equation}
and perturbation term $\bm{\pi}\!\!=\!\![\bm{0}, -\hat{\bm{M}}^{-1}\tilde{\bm{\tau}}]$. For notational brevity, we have substituted $\nu_{\bm{e}}$ and $\nu_{\dot{\bm{e}}}$ from \eqref{implct_eq_alpha} here. Now, we directly use $\bm{A}$ and the full state $\bm{x}$ to find a metric \eqref{metric_lyap} for which %the LMI
\begin{equation}
\bm{\mathcal{F}} \coloneqq \bm{A}^\top\bm{\mathcal{M}} + \bm{\mathcal{M}}\bm{A} + \hspace{0.1cm}\dot{\hspace{-0.1cm}\bm{\mathcal{M}}} \preceq -2\alpha\bm{\mathcal{M}} \label{lmi_ct}
\end{equation}
holds uniformly $\forall t$. Inspecting \eqref{A_pi}, we observe the benefit of multiplying with $\hat{\bm{M}}$ in the second line for reducing the complexity of satisfying \eqref{lmi_ct}. Therefore, \eqref{coord_metric} is postulated as a fitting metric, where we have introduced a stiffness $\bm{\mathcal{K}}$, which will be fully specified later on. In order to prove the result, we combine \eqref{lmi_ct} from \cite[Theorem 2]{lohmiller_ct} with \cite[Theorem 2.1]{corless}. Thus, computing the derivative of \eqref{V_ct_relatship} w.r.t. to~\eqref{coord_metric} gives%along \eqref{error_dyn} gives
\begin{equation}
\dot{V} = \tfrac{1}{2} \bm{x}^\top \bm{\mathcal{F}} \bm{x} + \bm{x}^\top\bm{\mathcal{M}}\bm{\pi} = \tfrac{1}{2} \bm{x}^\top \bm{\mathcal{F}} \bm{x} - (\varepsilon\bm{e}+\dot{\bm{e}})^\top\tilde{\bm{\tau}} \notag
\end{equation}
Since $\tilde{\bm{\tau}}$ is a time-varying disturbance independent of the error states $\bm{x}$, we need to upper bound the second term to bring the entire expression into the form $\dot V\leq -2\alpha(V-\underline{V})$ with constant $\underline{V}>0$ such that \cite[Theorem 2.1]{corless} is applicable. Using the Cauchy-Schwarz, triangle and Young's inequalities,
\begin{align}
- (\varepsilon\bm{e}+\dot{\bm{e}})^\top\tilde{\bm{\tau}} &\leq (|\varepsilon|\lVert\bm{e}\rVert+\lVert\dot{\bm{e}}\rVert)\lVert\tilde{\bm{\tau}}\rVert \notag\\
&\leq \tfrac{|\varepsilon|\vartheta}{2}\lVert\bm{e}\rVert^2 + \tfrac{\varphi}{2}\lVert\dot{\bm{e}}\rVert^2 + (\tfrac{|\varepsilon|}{\vartheta} + \tfrac{1}{\varphi}) \tfrac{\lVert\tilde{\bm{\tau}}\rVert^2}{2}
\end{align}
follows for positive constants $\vartheta,\varphi\in\mathbb{R}^+$. Next, we include the first two summands into the quadratic expression w.r.t.~$\bm{\mathcal{F}}$ and apply \cite[Lemma 2]{my_tro_article} for $\beta\sqrt{\underline{\lambda}(\bm{\Sigma}_{\bm{\tau}}(t))+s}+Ld \leq \Delta\sqrt{\alpha(t)}$ to eliminate the time-dependency of the remaining bias. Thus,
\begin{equation}
\mathrm{Pr}\!\left\{\dot{V} \leq \tfrac{1}{2} \bm{x}^\top (\bm{\mathcal{F}} \!+\! |\varepsilon|\vartheta\bm{I} \!\oplus\! \varphi\bm{I}) \bm{x} \!+\! \tfrac{\alpha}{2} (\tfrac{|\varepsilon|}{\vartheta} \!+\! \tfrac{1}{\varphi})\Delta^2\right\} = 1-\delta \notag \, ,
\end{equation}
which is equivalent to requiring for \cite[Theorem 2.1]{corless} that 
\begin{equation}
\mathrm{Pr}\Big\{\dot{V} \leq -2\alpha(t) \big[V - \big(\tfrac{|\varepsilon|}{\vartheta} + \tfrac{1}{\varphi}\big)\tfrac{\Delta^2}{4}\big] \Big\} = 1-\delta \label{prob_exp_conv_prf}
\end{equation}
given that the matrix inequality
\begin{equation}
\bm{A}^\top\bm{\mathcal{M}} + \bm{\mathcal{M}}\bm{A} + \hspace{0.1cm}\dot{\hspace{-0.1cm}\bm{\mathcal{M}}} + |\varepsilon|\vartheta\bm{I} \oplus \varphi\bm{I} \preceq -2\alpha(t)\bm{\mathcal{M}} \label{lmi_ct_ball}
\end{equation}
is fulfilled deterministically. Here, note that we make use of an extension of \cite[Theorem 2.1]{corless} based on Grönwall's inequality to include a more precise time-variant convergence rate $\alpha(t)\geq\underline{\alpha}\in\mathbb{R}^+$. Therefore, analyzing under which conditions \eqref{lmi_metric} and \eqref{lmi_ct_ball} hold is all that is left to prove exponential stability and convergence to the ball \eqref{th1_ball}--\eqref{prob_theo1}. The former is necessary for expressing the ball's radius as
\begin{equation}
\varrho^2(t) = \tfrac{2\underline{V}}{\underline{\mu}(t)} = \tfrac{|\varepsilon|/\vartheta + 1/\varphi}{2\underline{\mu}(t)}\Delta^2 \, , \label{radius_rho_sq}
\end{equation}
for which Lemma~\ref{lemma_met} or Corollary~\ref{coro_met} can be used, depending on the form of $\bm{\mathcal{K}}$. Proceeding with the latter condition, we plug \eqref{coord_metric} and \eqref{A_pi} into \eqref{lmi_ct_ball}, reformulate, and obtain %for $\varepsilon>0$ that
\begin{align}
\bm{0} &\preceq \begin{bmatrix} \underbrace{\varepsilon\bm{P}^\prime \!-\! \alpha\bm{\mathcal{K}} \!-\! \tfrac{\dot{\bm{\mathcal{K}}}}{2} \!-\! \tfrac{|\varepsilon|\vartheta}{2}\bm{I}}_{=\bm{P}} & \underbrace{\tfrac{1}{2}\bm{B}^\prime \!-\! \tfrac{\varepsilon}{2}\hat{\bm{C}}^\top}_{=\bm{B}} \\ \underbrace{\tfrac{1}{2}\bm{B}^\prime \!-\! \tfrac{\varepsilon}{2}\hat{\bm{C}}}_{=\bm{B}^{\top}} & \underbrace{\bm{E}^\prime \!-\! (\varepsilon\!+\!\alpha)\hat{\bm{M}} \!-\! \tfrac{\varphi}{2}\bm{I}}_{=\bm{E}} \end{bmatrix} \label{full_mat_ct}\\
\bm{P}^\prime &= \tfrac{\nu_{\bm{e}}}{\|\bm{e}\|^2}\bm{I} + \hat{\bm{K}}_G + \bm{K}_P\, , \quad \bm{E}^\prime = \tfrac{\nu_{\dot{\bm{e}}}}{\|\dot{\bm{e}}\|^2}\bm{I} + \hat{\bm{D}} + \bm{K}_D \, , \notag \\
\bm{B}^\prime &= \bm{P}^\prime-\bm{\mathcal{K}} + \varepsilon (\bm{E}^\prime - 2\alpha\hat{\bm{M}}) \, . \notag
\end{align}
Leveraging Schur complements~\cite{boyd_lmis} and symmetric eigenvalue principles~\cite{bhatia_ma}, we can derive the sufficient conditions
\begin{subequations}
\begin{align}
\bm{P} &\succ \bm{0} \Rightarrow \bm{E} \succeq \bm{B}^{\top} \bm{P}^{-1}\bm{B} \Leftrightarrow \underline{\lambda}(\bm{E}) \geq \tfrac{[\bar{\lambda}(\bm{B}^\prime) + |\varepsilon|\lVert \hat{\bm{C}} \rVert ]^2}{4\underline{\lambda}(\bm{P})} \notag \\
\bm{E} &\succ \bm{0} \Rightarrow \!\bm{P} \succeq \bm{B}\bm{E}^{-1}\bm{B}^{\top} \Leftrightarrow \underline{\lambda}(\!\bm{P}) \geq \tfrac{[\bar{\lambda}(\bm{B}^\prime) + |\varepsilon|\lVert \hat{\bm{C}} \rVert ]^2}{4\underline{\lambda}(\bm{E})} \notag
\end{align}
\end{subequations} % we derive the sufficient conditions,
which are unified compactly to $\bm{P},\bm{E} \succ \bm{0}$ and $\bm{B}^\prime \succ \bm{0}$ with 
\begin{equation}
4\underline{\lambda}(\bm{P})\underline{\lambda}(\bm{E}) \geq [\bar{\lambda}(\bm{B}^\prime) + |\varepsilon|\lVert \hat{\bm{C}} \rVert ]^2 \, . \notag%\label{eig_val_ac}
\end{equation}
In order to maximize the convergence rate $\alpha$ and region $\bar{\varrho}$, we thus in turn aim to maximize $\underline{\lambda}(\bm{P}),\underline{\lambda}(\bm{E})$ while minimizing $\bar{\lambda}(\bm{B}^\prime)$ subject to $\underline{\lambda}(\bm{B}^\prime)>0$. Therefore, we choose
\begin{equation}
\bm{\mathcal{K}} = \hat{\bm{\mathcal{K}}}_G+\kappa\bm{I} \label{stiffness_metric}
\end{equation}
with constant $\kappa\in\mathbb{R}^+$ and split up \eqref{full_mat_ct} for $\eta_{\bm{e}}\!\coloneqq\!\tfrac{\nu_{\bm{e}}}{\|\bm{e}\|^2}$ into
\begin{align}
\bm{0} \preceq& \underbrace{\begin{bmatrix} a\bm{I}  & \tfrac{b}{2}\bm{I} \!-\! \varepsilon\alpha\hat{\bm{M}} \\ \tfrac{b}{2}\bm{I} \!-\! \varepsilon\alpha\hat{\bm{M}} & \gamma\bm{I} \!-\! (\varepsilon\!+\!\alpha)\hat{\bm{M}} \end{bmatrix}}_{=\bm{\Upsilon}(\alpha)} \! + \!\begin{bmatrix} \tilde{\bm{G}}\!+\!\varepsilon\eta_{\bm{e}}\bm{I} & \tfrac{\eta_{\bm{e}}+\varepsilon\eta_{\dot{\bm{e}}}}{2}\bm{I} \\ \tfrac{\eta_{\bm{e}}+\varepsilon\eta_{\dot{\bm{e}}}}{2}\bm{I} & \eta_{\dot{\bm{e}}}\bm{I} \end{bmatrix} \notag\\
& \!+\! \underbrace{\begin{bmatrix} \varepsilon\tilde{\bm{K}}_P & \tfrac{1}{2}(\tilde{\bm{K}}_P \!+\! \varepsilon(\tilde{\bm{D}} \!-\! \hat{\bm{C}}^\top ) ) \\ \tfrac{1}{2}(\tilde{\bm{K}}_P \!+\! \varepsilon(\tilde{\bm{D}} \!-\! \hat{\bm{C}} ) )  & \tilde{\bm{D}} \end{bmatrix}}_{=\bm{R}} , \label{decomp_mat_ct}\\ %\!+\! \underbrace{\begin{bmatrix} \bm{0} \!\! & \tfrac{\varepsilon}{2}\tilde{\bm{D}} \\ \tfrac{\varepsilon}{2}\tilde{\bm{D}} \!\! & \tilde{\bm{D}} \end{bmatrix}}_{=\Delta\bm{\mathcal{D}}} \!- \tfrac{\varepsilon}{2} \! \underbrace{\begin{bmatrix} \bm{0} \!\! &  \\ \hat{\bm{C}} \!\! & \bm{0} \end{bmatrix}}_{=\bm{\Gamma}} 
\tilde{\bm{G}} \!=& \varepsilon\hat{\bm{K}}_G \!-\! \alpha\hat{\bm{\mathcal{K}}}_G, \quad \tilde{\bm{K}}_P \!=\!\! \bm{K}_P \!-\! \underline{k}_P\bm{I}, \quad \tilde{\bm{D}} \!=\!\! \hat{\bm{D}} \!-\! \underline{\hat{d}}\bm{I} \!+\! \tilde{\bm{K}}_D . \notag
\end{align}%\tfrac{\eta_{\bm{e}}}{2}\!+\! \tfrac{\eta_{\dot{\bm{e}}}}{2}\!+\!
Here, we have minimized $\bm{B}^\prime$ and eliminated $-\tfrac{1}{2}\dot{\hat{\bm{\mathcal{K}}}}_G$ in $\bm{P}$ by cancelling $\hat{\bm{K}}_G$ and exploiting $\hat{G}(\bm{e}) = \tfrac{1}{2}\bm{e}^\top\hat{\bm{\mathcal{K}}}_G\bm{e}$ with
\begin{equation}
\dot{\hat{G}}(\bm{e})=\dot{\bm{e}}^\top\hat{\bm{K}}_G\bm{e} \quad\Rightarrow\quad \tfrac{1}{2}\bm{e}^\top\dot{\hat{\bm{\mathcal{K}}}}_G\bm{e} = \dot{\bm{e}}^\top(\hat{\bm{K}}_G-\hat{\bm{\mathcal{K}}}_G)\bm{e} \, . \notag
\end{equation}
Also, we have directly made use of the obvious requirement $\varepsilon>0$ for $\bm{P} \succ \bm{0}$ in \eqref{full_mat_ct} and introduced the constants
\begin{equation}
a \!=\! \varepsilon(\underline{k}_P\!-\!\tfrac{\vartheta}{2})\!-\! \alpha\kappa , \quad b \!=\! \underline{k}_P \!+\! \varepsilon(\underline{\hat{d}}\!+\!\underline{k}_D) \!-\! \kappa , \quad \gamma \!=\! \underline{\hat{d}}\!+\!\underline{k}_D\!-\! \tfrac{\varphi}{2}. \notag
\end{equation}
Note that, for the sake of simplicity, we refrain from an extended choice of $\bm{\mathcal{K}}$ enforcing $\bm{B}^\prime=\bm{0}$ since it would involve additional matrix derivatives $\dot{\hat{\bm{D}}},\dot{\bm{K}}_{P,D},(\dot{\eta}_{\bm{e}}\!+\!\varepsilon\dot{\eta}_{\dot{\bm{e}}})\bm{I}$. Therefore, we leverage Weyl's inequalities \cite[III.2.1]{bhatia_ma} next to transform~\eqref{full_mat_ct} using the symmetric decomposition \eqref{decomp_mat_ct} to
\begin{equation}
\alpha\hat{G} \leq [\underbrace{\underline{\lambda}(\bm{\Upsilon}(\alpha)) \!+\! \underline{\lambda}(\bm{R})}_{=\lambda(\alpha,t)}]\lVert\bm{x}\rVert^2 \!+\! \varepsilon(\bm{e}^\top\!\hat{\bm{g}}_{\bm{e}}  \!+\! \nu_{\bm{e}} \!+\! \omega_{\dot{\bm{e}}}) \!+\! \nu_{\dot{\bm{e}}} \!+\! \omega_{\bm{e}} \label{min_eig_val_cond_upsilon}
\end{equation}%lower left and right
with $\nu_{\bm{e}}$, $\nu_{\dot{\bm{e}}}$ and $\omega_{\bm{e}}$, $\omega_{\dot{\bm{e}}}$ from \eqref{implct_eq_alpha}. Now, we exploit the commutability of the submatrices of $\bm{\Upsilon}$ in \eqref{decomp_mat_ct} to apply \cite[Theorem 3]{silvester_block_det} and follow for $\bm{\upsilon}\!=\![a\!+\!\gamma\!-\!(\varepsilon\!+\!\alpha)\hat{\bm{m}}]/2$ that %the minimal eigenvalues of all summands exactly as
\begin{align}
\bm{\lambda}(\bm{\Upsilon}(\alpha)) \!=\,& \bm{\upsilon}(\alpha) \!\pm\! \sqrt{\bm{\upsilon}^2(\alpha) \!+\! \big(\varepsilon\alpha\hat{\bm{m}} \!-\! \tfrac{b}{2}\big)^2 \!-\! a\big[\gamma\!-\!(\varepsilon\!+\!\alpha)\hat{\bm{m}}\big]} \notag \\
=\,& \bm{\upsilon}(\alpha) \!\pm\! \sqrt{\tfrac{\left[\gamma-a(\alpha)-(\varepsilon+\alpha)\hat{\bm{m}}\right]^2}{4} \!+\! \big(\varepsilon\alpha\hat{\bm{m}}\!-\!\tfrac{b}{2}\big)^2} . \label{eig_upsilon} 
\end{align}
To obtain an optimal parameterization, we consider
\begin{equation}
\max_{\kappa,\varphi} \underline{\lambda}(\bm{\Upsilon}) = \varepsilon(\underline{k}_P\!-\!\tfrac{\vartheta}{2})\!-\! \alpha\kappa \!-\! \tfrac{(\varepsilon+\alpha)\Delta\hat{m}}{4} \left(1\!+\!\sqrt{1\!+\!\tfrac{4\varepsilon^2\alpha^2}{(\varepsilon+\alpha)^2}}\right), \notag%\label{opt_upsilon} %\quad \mathrm{s.t.}~\, 0 < \varepsilon < \sqrt{\tfrac{\kappa}{\underline{\hat{m}}}} \,,\notag 
\end{equation}
which then, subsequently, allows a maximal choice of convergence rate $\alpha$ such that \eqref{min_eig_val_cond_upsilon} and thus \eqref{lmi_ct_ball} hold with equality. Here, we have directly reduced $\max\underline{\lambda}(\bm{\Upsilon})$ by optimally centering the two spheres generated by the spectrum of $\hat{\bm{M}}$ in the origin via the requirements $b=\varepsilon\alpha(\underline{\hat{m}}+\bar{\hat{m}})$ and $\gamma = a + (\varepsilon+\alpha)(\underline{\hat{m}}+\bar{\hat{m}})/2$, leading to the optimal parameters \eqref{opt_kappa}--\eqref{opt_phi}.
Then, given values $\varepsilon,\vartheta\in\mathbb{R}^+$, obtained, e.g., via further numerical optimization based on the radius \eqref{radius_rho_sq}, we plug \eqref{eig_upsilon} into \eqref{min_eig_val_cond_upsilon} and solve for the maximal $\alpha$ such that the resulting quadratic expression $0\leq a_0-2a_1\alpha+a_2\alpha^2$ holds with equality. Finally, we obtain the time-variant rate $\alpha(t)$ as
\begin{equation}
\alpha(t) = \tfrac{1}{a_2(t)} \left( a_1(t) - \sqrt{a_1^2(t) - a_0(t)a_2(t)}\right) \,, \label{alpha_t_quad_sol}
\end{equation}
with $a_{1-3}(t)$ given for $\hat{m}^*(t)$ s.t.~$\underline{\lambda}(\bm{\Upsilon}(\alpha))=\lambda^{-}(\bm{\Upsilon}(\hat{m}^*))$ by
\begin{align}
a_0(t) &= \xi^2(t) - \zeta^2(t) - b^2(t) \,, \notag \\
a_1(t) &= \varkappa(t)\xi(t) + [\kappa-\hat{m}^*(t)]\zeta(t) - 2\varepsilon\hat{m}^*(t)b(t) \,, \notag\\
a_2(t) &= \varkappa^2(t) - [\kappa-\hat{m}^*(t)]^2 - [2\varepsilon\hat{m}^*(t)]^2 \,, \notag
\end{align}
where $\varkappa(t)  \!=\! \kappa \!+\! \hat{m}^* \!\!\!+\!\! \tfrac{2\hat{G}(\bm{e})}{\|\bm{e}\|^2\!+\!\|\dot{\bm{e}}\|^2}$, $\zeta(t) \!=\! \gamma \!-\! \varepsilon(\underline{k}_P\!-\!\tfrac{\vartheta}{2}\!+\!\hat{m}^*)$ and 
\begin{equation}
\xi(t) \!=\! \varepsilon(\underline{k}_P\!-\!\tfrac{\vartheta}{2}\!-\!\hat{m}^*\!+\!2\tfrac{\bm{e}^\top\!\hat{\bm{g}}_{\bm{e}} \!+\! \nu_{\bm{e}} \!+\! \omega_{\dot{\bm{e}}}}{\|\bm{e}\|^2\!+\!\|\dot{\bm{e}}\|^2}\!) \!+\! 2[\!\tfrac{\nu_{\dot{\bm{e}}} + \omega_{\bm{e}}}{\|\bm{e}\|^2\!+\!\|\dot{\bm{e}}\|^2} \!+\! \underline{\lambda}(\bm{R})] \!+\! \gamma \notag .
\end{equation}
All together, exponential convergence and stability to within the ball $B(\varrho)$ from \eqref{th1_ball} can thus be concluded in the region
\begin{align}
\bm{\mathcal{E}} = \big\{ &\bm{e},\dot{\bm{e}}\in\mathbb{R}^N \,\big|\, \underline{\alpha}\hat{G}(\bm{e}) \leq [\underline{\lambda}(\bm{\Upsilon}) \!+\! \underline{\lambda}(\bm{R})](\lVert\bm{e}\rVert^2 \!+\! \lVert\dot{\bm{e}}\rVert^2) \notag\\
& \hspace{3cm}\!+\! \varepsilon(\bm{e}^\top\!\hat{\bm{g}}_{\bm{e}}  \!+\! \nu_{\bm{e}} \!+\! \omega_{\dot{\bm{e}}}) \!+\! \nu_{\dot{\bm{e}}} \!+\! \omega_{\bm{e}} ,\notag\\
& \beta(\delta,d)\sqrt{\underline{\lambda}(\bm{\Sigma}_{\bm{\tau}}(t))+s(d)}+Ld \leq \Delta\sqrt{\alpha(t)}, \, \forall t\geq t_0 \big\} \notag
\end{align}
since \eqref{alpha_t_quad_sol} guarantees that \eqref{min_eig_val_cond_upsilon} holds, and thus also \eqref{prob_exp_conv_prf} due to \eqref{lmi_ct_ball}--\eqref{decomp_mat_ct}. $\hfill\blacksquare$
\end{proof_th}

\subsection{Design Interpretations}

Intuitively, Lemma~\ref{lemma_met} provides a systematic approach to designing a Lyapunov function (or a contraction metric) for exponential stability proofs applied to Lagrangian systems. Typically, $\varepsilon$ is chosen sufficiently small to ensure negative definiteness of the derivative, see, e.g., \cite[p.~187]{murray_amitrm}, and eigenvalues are conservatively bounded, whereas we enable optimization of $\varepsilon$ on the interval \eqref{eps_intrvl_theo1} and provide exact expressions for the metric's spectrum in Corollary~\ref{coro_met}. Moreover, the results allow a systematic design of the stiffness $\bm{\mathcal{K}}$ contrary to the common approach only using the system's potential energy. In our case, we have employed the suboptimal choice \eqref{stiffness_metric} in sight of the analytical simplicity of the result, for which an optimal parametrization \eqref{opt_kappa} is then derived. Furthermore, we observe from \eqref{full_mat_ct} in the proof of Theorem~\ref{lgp_nat_pdp_theo1} that setting $\bm{\mathcal{K}} = \bm{P}^\prime + \varepsilon(\bm{E}^\prime-2\alpha\hat{\bm{M}})$ would maximize the spectrum of the considered matrix inequality on the right-hand side, and thus also the admissable convergence rate $\alpha$.

The benefit of our structure-preserving control design is emphasized theoretically in Theorem~\ref{lgp_nat_pdp_theo1}. It is particularly visible in the influences on the exponential convergence rate in \eqref{min_eig_val_cond_upsilon}. Here, it becomes clear that the preservation of the potential energy \eqref{des_pot_frcs} and the conservative projector-based compensation \eqref{c_law}--\eqref{proj_mat}, associated in \eqref{min_eig_val_cond_upsilon} with the terms $\bm{e}^\top\!\hat{\bm{g}}_{\bm{e}}$ and $\nu_{\bm{e}}\!\!=\!\!h(\bm{e}^\top\!\hat{\bm{g}}_{\bm{q}})\bm{e}^\top\!\hat{\bm{g}}_{\bm{q}}$, $\nu_{\dot{\bm{e}}}\!\!=\!\!h(\dot{\bm{e}}^\top\!\hat{\bm{d}}_{\dot{\bm{q}}})\dot{\bm{e}}^\top\!\hat{\bm{d}}_{\dot{\bm{q}}}$, respectively, have a beneficial impact on the stability of the system. Moreover, the radius of the ball \eqref{th1_ball} clearly decreases with the term $\hat{G}(\bm{e})$ in \eqref{mu_lb_th1} due to the structural preservation of the potential energy in the closed-loop system.

\begin{figure*}[]
\includegraphics[width=\textwidth]{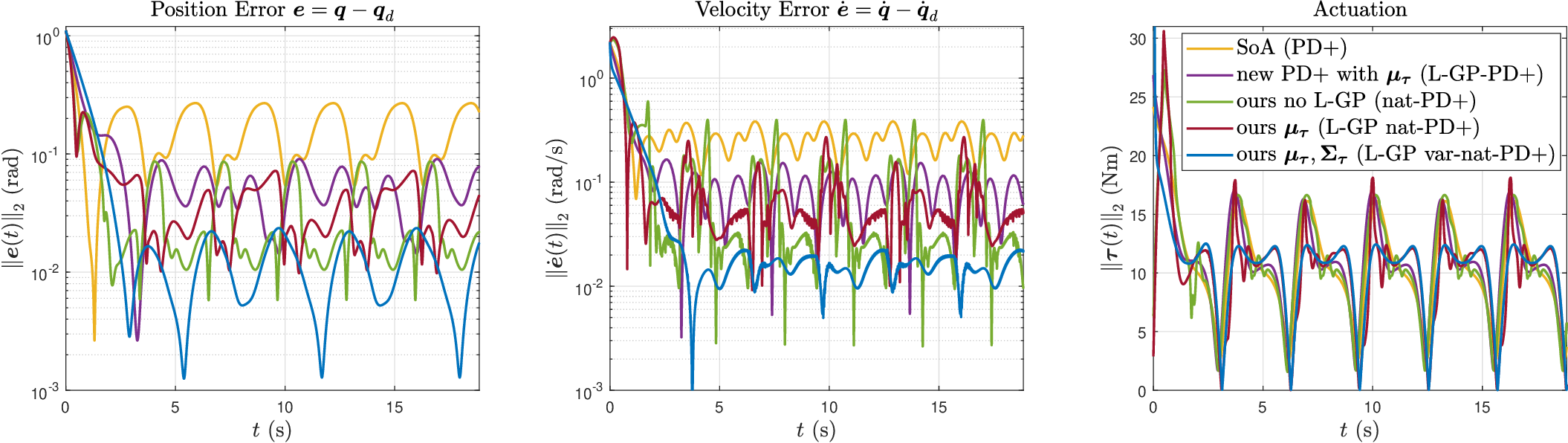}
\caption[]{Tracking performances of the standard parametric and proposed, structure-preserving or L-GP-based, PD+ controllers.}
\label{log_errors_tl}
\end{figure*}

Also, note that a similar exponential stability result for the standard PD+ controller \cite{pdp_ref} can be derived using an ablation of the argumentation in Theorem~1. Then, the desired energy quantities become $\bm{g}_d(\bm{e})\!\!=\!\!\bm{K}_P\bm{e}$, $\bm{d}_d(\dot{\bm{e}})\!\!=\!\!\bm{K}_D\dot{\bm{e}}$ with constant non-adaptive $\bm{K}_{P,D}\!\!\in\!\!\mathbb{R}^{N\times N}$, and friction as well as gravity are fully compensated. Therefore, we can conclude, based on our theoretical investigations, that the classical controller \cite{pdp_ref} is less robust due to the lower convergence rate and higher radius, stemming from the missing structural benefits since gravity and friction are fully compensated. This becomes particularly important for handling the possibly destabilizing effect of the Coriolis cross-term in \eqref{decomp_mat_ct} since $\bm{R}$ consists in the worst-case of a negative summand whose impact is scaled with $\lVert\hat{\bm{C}}\rVert\!\leq\!(\hat{c}_0 \!+\! \hat{c}_1\lVert\bm{q}\rVert)\lVert\dot{\bm{q}}\rVert$ with positive constants $\hat{c}_0,\hat{c}_1\!\in\!\mathbb{R}^+$.

Furthermore, the positive effect of the variance-based gain adaptation \eqref{sigma_gain_adapt_law} can be seen by inspecting $\underline{\lambda}(\bm{R})$ from \eqref{decomp_mat_ct}--\eqref{min_eig_val_cond_upsilon}, since it has a direct influence on the convergence rate $\alpha(t)$, c.f., e.g., \eqref{implct_eq_alpha}, via $\lambda(\alpha,t)\!\!=\!\!\underline{\lambda}(\bm{\Upsilon})\!\!+\!\!\underline{\lambda}(\bm{R})$. Reformulating $\bm{R}$ from \eqref{decomp_mat_ct} for $\tilde{\bm{K}}_P\!=\!\tilde{\bm{K}}_D$ to
\begin{equation}
\bm{R} = \underbrace{ \begin{bmatrix} \varepsilon\tilde{\bm{K}}_P & \tfrac{1+\varepsilon}{2}\tilde{\bm{K}}_P \\ \tfrac{1+\varepsilon}{2}\tilde{\bm{K}}_P & \tilde{\bm{K}}_P \end{bmatrix} }_{=\tilde{\bm{K}}} + \underbrace{ \begin{bmatrix} \bm{0} & \tfrac{\varepsilon}{2}\tilde{\bm{D}} \\ \tfrac{\varepsilon}{2}\tilde{\bm{D}} & \tilde{\bm{D}} \end{bmatrix} }_{=\tilde{\bm{\mathcal{D}}}} - \tfrac{\varepsilon}{2} \begin{bmatrix} \bm{0} & \hat{\bm{C}}^\top \\ \hat{\bm{C}} & \bm{0} \end{bmatrix} \notag ,
\end{equation}
we follow with Weyl's inequalities \cite[III.2.1]{bhatia_ma} that $\underline{\lambda}(\bm{R}) \!\geq\! \underline{\lambda}(\tilde{\bm{K}}) \!+\! \underline{\lambda}(\tilde{\bm{\mathcal{D}}}) \!-\! \tfrac{\varepsilon}{2}\|\hat{\bm{C}}\|$. Computing the spectrum %of $\tilde{\bm{K}}$ according to
\begin{equation}
\bm{\lambda}(\tilde{\bm{K}}) = \tfrac{\bm{\lambda}(\tilde{\bm{K}}_P)}{2} \left(1+\varepsilon\pm\sqrt{2(1+\varepsilon^2)}\right) \notag ,
\end{equation}
where we have applied \cite[Theorem 3]{silvester_block_det}, we follow that $\underline{\lambda}(\tilde{\bm{K}})\!\geq\!0$ holds beneficially for $\varepsilon \!\geq\! 1$.

In the following numerical evaluations, we will validate these observations and even show that a classical PD+ controller can become unstable very early while our control law remains robust despite increasing disturbance speeds. 

\section{Numerical Illustrations}\label{num_sec}

In this section, we validate the efficacy of our proposed methods in numerical simulations.\footnote{\textbf{Code:} For reproduction of our experiments, Matlab code along with an L-GP toolbox are available under: \href{https://github.com/gevangelisti/lgp_prjctr_ctrl}{https://github.com/gevangelisti/lgp\_prjctr\_ctrl}} At first, we start with a simple benchmark example for the sake of accessible interpretation and comparability and then move on to a higher system complexity illustrating scalable applicability along with practical feasibility.

\subsection{Two-link Manipulator}\label{two_link_subsection}

\subsubsection{Setup}

We benchmark our proposed control methods using the two-link robotic manipulator from \cite[p.~164]{murray_amitrm}. Gravity $g\!\!=\!\!10\!$ m/s$^2$ acts along the positive $x$-axis as in~\cite{mein_cdc_paper} such that $\bm{q}\!\!=\!\!\bm{0}$ is an equilibrium of the system. The links are parametrized to have unit masses $m_n\!\!=\!\!1$~kg and lengths $l_n\!\!=\!\!1$~m for $n\!\!\in\!\!\{1,2\}$. We additionally include dissipation via unit linear and quadratic damper elements at each joint, leading to $\bm{D}(\dot{\bm{q}})\!\!=\!\!d_1\bm{I}\!\!+\!\!d_2\mathrm{diag}(|\dot{\bm{q}}|)$ with $d_1\!\!=\!\!d_2\!\!=\!\!1$. Parameter estimates are available but erroneous: $\hat{m}_n\!\!=\!\!(1\!\!+\!\!\chi_n)m_n$, $\hat{l}_n\!\!=\!\!(1\!\!+\!\!\chi_n)l_n$ and $\hat{d}_n\!\!=\!\!(1\!\!-\!\!\chi_n)d_n$ with an alternating relative bias of $50\%$ such that $\chi_n\!\!=\!\!(-1)^{n-1}/2$.

\begin{figure}
\centering
\includegraphics[width=\linewidth]{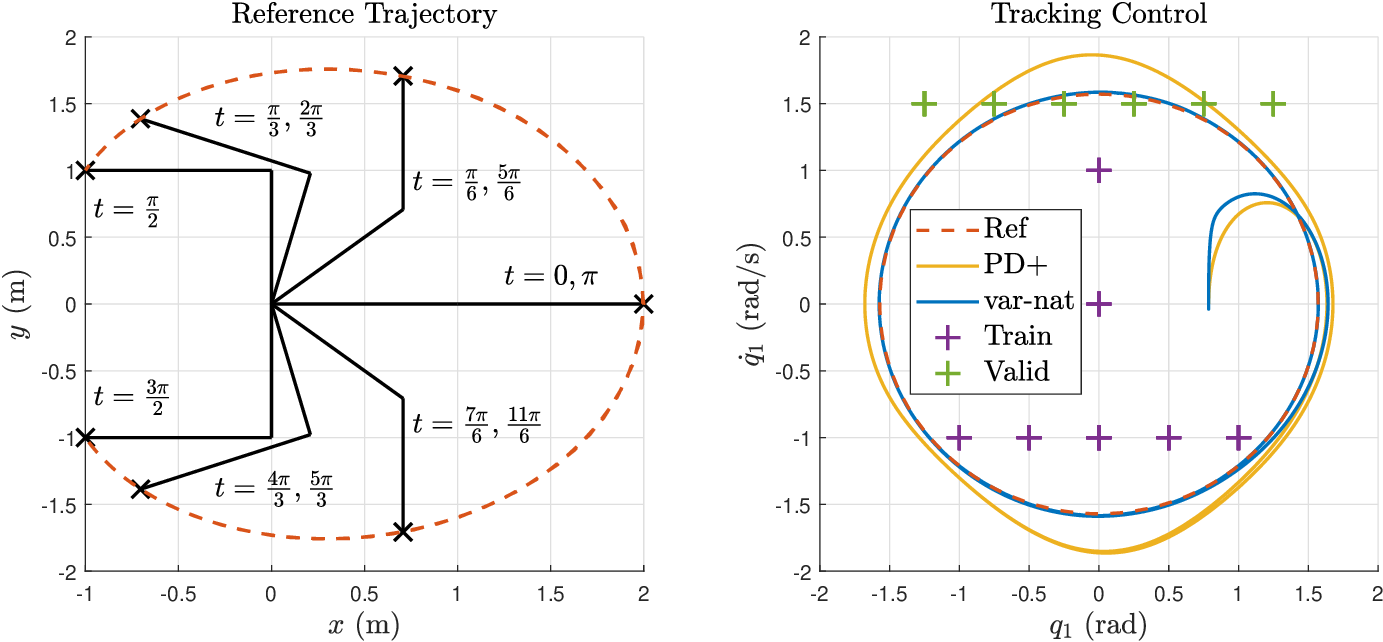}
\caption[PD+ control]{Reference trajectory in the cartesian workspace of the two-link (left) and tracking performance comparison of the standard parametric with the proposed variance-adaptive and structure-preserving, L-GP-based, PD+ controller (right).} 
\label{two_link_pdp_comp}
\end{figure}%

For the L-GP's hyperparameters, we use $D\!\!=\!\!5^2\!+\!3^2$ training and $V\!\!=\!\!6^2$ validation data pairs stemming from equally distanced grids on the domains $\bm{q}\!\!\in\!\![-a,a]^2$ for $a\!\!=\!\!1,\!1.25$ with fixed velocity $\dot{q}_n\!\!=\!\!(-1)^{n-1}\!\!,\!(2\!-\! n)1.5$ and acceleration $\ddot{\bm{q}}\!\!=\!\!\bm{4},\!\bm{0}$, respectively. Due to the additional dissipative subcomponent, the training data set also includes a small $3\!\!\times\!\!3$ grid in the velocity domain $\dot{\bm{q}}\!\!\in\!\![-1,1]^2$ for $\bm{q}\!=\!\ddot{\bm{q}}\!=\!\bm{0}$. Torque and acceleration measurements are corrupted by i.i.d.~noise with standard deviations of $0.1$ Nm and $\pi/180$ $\text{rad}/\text{s}^{2}$, respectively. We reduce the kinetic mass inertia hypermetric to the constant Euclidian form $\bm{\Lambda}^{-1}\!\!=\!\!\sigma_{\bm{d}_{T}}^{2}\bm{I}$. Similarly, we assume a  gravitational distance covariance $\bm{\Sigma}_{\bm{d}_G}\!\!=\!\!\mathrm{diag}(\bm{\sigma}_{\bm{d}_G}^{2})$. The hyperparameters are then optimized via the least-squares approximation of all $D\!+\!V\!=\!70$ measurements.

\subsubsection{Control}

As a control task, we consider the goal of tracking a sinusoidal reference trajectory $\bm{q}_d(t)\!\!=\!\!\pi/2\sin(t)\bm{1}$ starting from the initial condition $\bm{q}(0)\!\!=\!\!\pi/4\bm{1}$ and $\dot{\bm{q}}(0)\!\!=\!\!\bm{0}$, cf. Fig.~\ref{two_link_pdp_comp}. In addition to our proposed structure-preserving and L-GP-based controllers, we also compare with the standard parametric PD+ ablation from \cite[p.~194]{murray_amitrm} along with its L-GP extended counterpart. To prevent numerical stability issues, we implement the required projection matrices as
\begin{equation}
\bm{P}_{\bm{e}} = \tfrac{\bm{e}\bm{e}^\top}{\epsilon+\lVert\bm{e}\rVert^2} \label{project_epsilon}
\end{equation}%\cite[Appendix~A.3]{rasmussen_gp_mit_book}
with an $\epsilon\!\!=\!\!10^{-3}$. Note that this is a common numerical procedure often employed in other inversion implementations such as, e.g., mass matrix predictions with DeLaNs \cite{lutter2021combining} or Cholesky decompositions of covariance matrices \cite{rasmussen_gp_mit_book}. Moreover, all controllers are parametrized by the same proportional and derivative corrections $\bm{K}_P\!\!=\!\!\bm{K}_D\!\!=\!\!10\bm{I}$. The variance-adaptive natural PD+ controller additonally uses the gains \eqref{sigma_gain_adapt_law}, which are identically set to $\bm{K}_i\!\!=\!\!k_i\bm{I}$ with $k_1\!\!=\!\!10^2$, $k_2\!\!=\!\!0.02$ and $k_3\!\!=\!\!7.11$ such that $\underline{k}_P\!\!=\!\!\underline{k}_D\!\!=\!\!10\!\!+\!\!1$. Note that instead of assuming their measurability, the accelerations required to compute the L-GP covariance are estimated based on the model.

\begin{table}[]
\begin{tabular}{l|ccccc}
& PD+    & L-GP-PD+ & nat-PD+ & \begin{tabular}[c]{@{}c@{}}L-GP\\ nat-PD+\end{tabular} & \begin{tabular}[c]{@{}c@{}}L-GP var-\\ nat-PD+\end{tabular} \\ \hline
$\lVert\bm{\tau}\rVert_{\mathcal{L}_2}$ & 35.62  & 35.71	& 36.23 & 33.60 & \textbf{33.28} \\
$\max(\lVert\bm{\tau}\rVert)$\hspace{-0.1cm} & 16.12 & 16.95 & 16.64 & 18.12 & \textbf{12.43} \\
$\mathrm{E}[\lVert\bm{\tau}\rVert]$ & 10.599 & 10.707 & 10.769 & \textbf{10.029} & 10.086 \\ 
$\lVert[\bm{e}^\top\! \dot{\bm{e}}^\top\!]\rVert_{\mathcal{L}_2}$\hspace{-0.2cm} & 0.989 & 0.352 & 0.390 & 0.258 & \textbf{0.066} \\
$\max(\lVert\bm{e}\rVert)$  & 0.269	& 0.091 & 0.086 & 0.072 & \textbf{0.024} \\
$\max(\lVert\dot{\bm{e}}\rVert)$  & 0.384 & 0.154 & 0.397 & 0.271 & \textbf{0.023} \\
$\mathrm{E}[\lVert\bm{e}\rVert]$ & 0.156 & 0.053 & 0.032 & 0.031 & \textbf{0.012} \\
$\mathrm{E}[\lVert\dot{\bm{e}}\rVert]$ & 0.251 & 0.089 & 0.080 & 0.061 & \textbf{0.015}                                                      
\end{tabular}
\caption{Numerical evaluation of the steady-state controller performances for $t\geq10$s based on the two-link. Lower values indicate better performances w.r.t.~the considered metrics.}
\label{table_num_vals_two_link_pdp_ss}
\end{table}

\begin{figure}
\centering
\includegraphics[width=\linewidth]{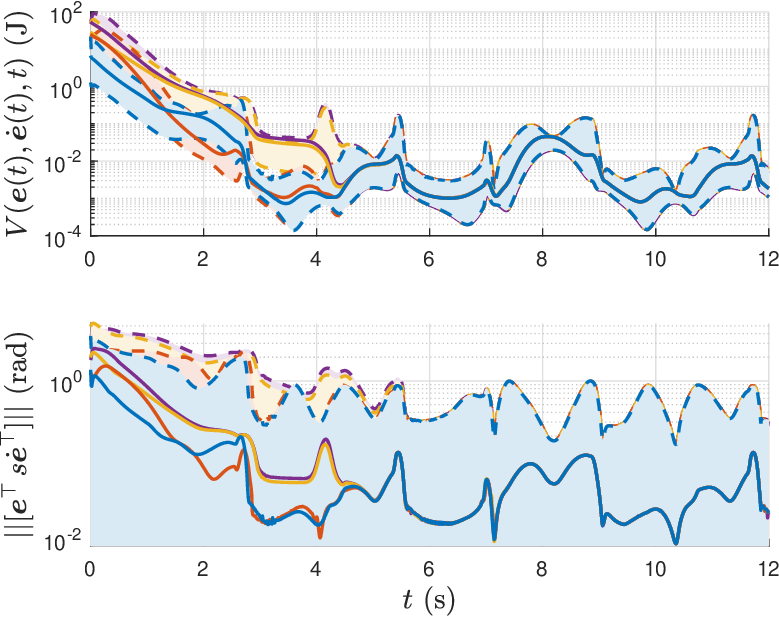}
\caption{Two-link: Lyapunov function and bounds for trajectories of the L-GP nat-PD+ controller with initial conditions $\bm{q}_0\!\sim\!\mathcal{N}(\bm{0},\sigma_0^2\bm{I})$ and $\dot{\bm{q}}_0\!\sim\!\mathcal{N}(\pi/2\bm{1},\sigma_0^2\bm{I})$, where $\sigma_0\!=\!\pi/3$. Solid lines indicate trajectories, shaded areas respective functional or error norm regions from Theorem~\ref{lgp_nat_pdp_theo1} and Lemma~\ref{lemma_met} with the bounds \eqref{prob_theo1} and \eqref{lmi_bounds_metric} given by the dashed lines.} 
\label{two_link_nat_pdp_V_and_norm_ede}
\end{figure}

\begin{figure}
\centering
\includegraphics[width=\linewidth]{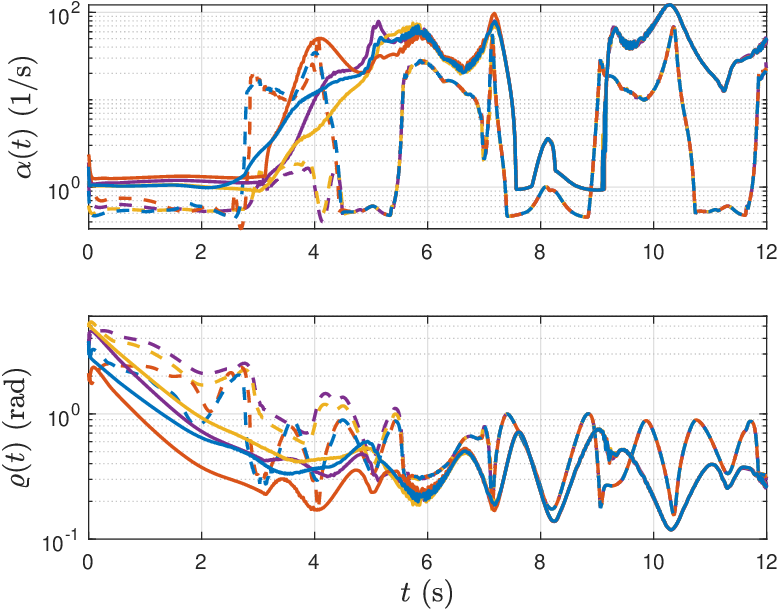}
\caption{Exponential convergence behaviors of the proposed natural structure-preserving controllers evaluated on the two-link for randomly drawn initial conditions as in Fig.~\ref{two_link_nat_pdp_V_and_norm_ede}. Solid lines indicate the variance-adaptive (L-GP var-nat-PD+) controller, dashed the static gain (L-GP nat-PD+) variant.} 
\label{two_link_exp_cov_params}
\end{figure}

\begin{figure*}[]
\centering
\includegraphics[width=\textwidth]{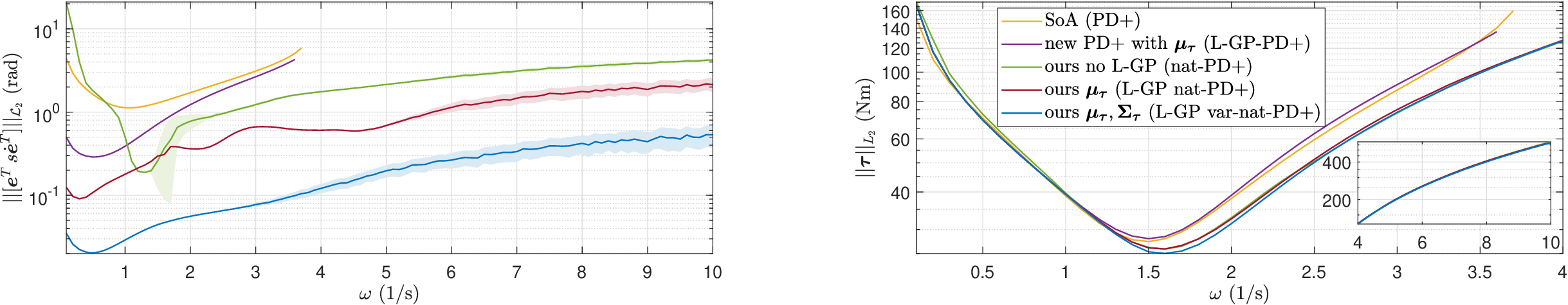}
\caption[]{Monte Carlo evaluation of the standard parametric and the proposed, natural structure-preserving or L-GP-based, PD+ controllers. Solid lines and shaded areas depict respective mean and standard deviation of $100$ realizations per frequency sample. Initial conditions are randomly drawn from the uniform distribution over $[\bm{q}_0^\top \dot{\bm{q}}_0^\top] \!\in\! [-\pi/4, \pi/4]^4$. Both standard PD+ controllers with full gravity cancellation become unstable at 3.7 and 3.8 rad s$^{-1}$.}
\label{tl_monte_carlo}
\end{figure*}

\subsubsection{Results}\label{two_link_subsub_results}

The simulation results are illustrated in Fig.~\ref{log_errors_tl} and the steady-state performances evaluated numerically in Tab.~\ref{table_num_vals_two_link_pdp_ss}. Despite considerable errors in the parameters, the proposed natural PD+ controller already shows significantly improved error metrics, demonstrating its natural robustness due to the preservation of the system's physical structure. Furthermore, employing the mean estimate of the L-GP greatly improves the tracking accuracy of both standard and natural PD+ controllers. Intuition for the functioning of the structure-preserving controllers can be gained by inspecting Fig.~\ref{log_errors_tl} for small transient times close to zero. Here, both variants without uncertainty-adaptation start with considerably lower actuation efforts due to the exploited system structures. The proposed uncertainty-based adaptation based on the L-GP's covariance matrix estimate, variably shaping and injecting potential energy and damping, respectively, demonstrates the highest tracking accuracy w.r.t.~all error metrics while having to actuate with the lowest effort based on the torque input's $\mathcal{L}_2$-norm. Its slightly increased mean over time hints at the trade-off between actuation and accuracy performed by the uncertainty-dependent feedback gains.

\begin{figure}[]
\centering
\includegraphics[width=\linewidth]{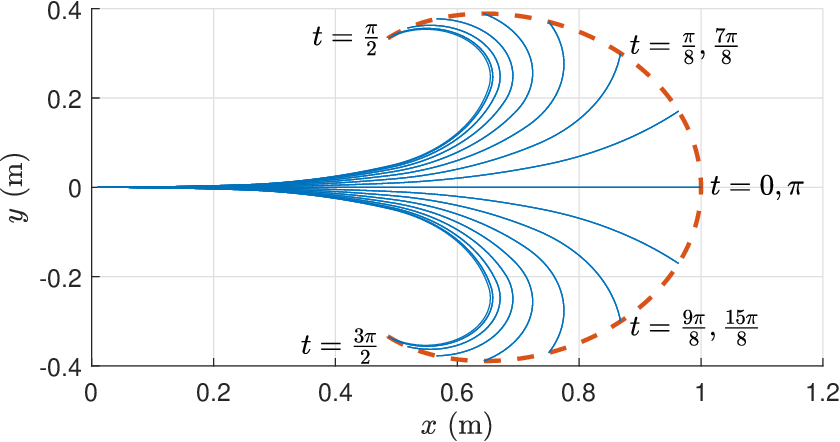}
\caption{Reference trajectory of the soft robot moving in its planar workspace.} 
\label{soro_ref_ws}
\end{figure}

For the validation of Theorem~\ref{lgp_nat_pdp_theo1}, we first evaluate the proposed Lyapunov function \eqref{V_ct_relatship} and its bounds \eqref{lmi_bounds_metric} in Fig.~\ref{two_link_nat_pdp_V_and_norm_ede} based on the metric~\eqref{coord_metric} from Lemma~\ref{lemma_met} with the design stiffness \eqref{stiffness_metric} for different closed-loop trajectories of the L-GP-based natural PD+ controller. Initial conditions are randomly drawn from $\bm{q}(0)\!\sim\!\mathcal{N}(\bm{0},\sigma_0^2\bm{I})$ and $\dot{\bm{q}}(0)\!\sim\!\mathcal{N}(\pi/2\bm{1},\sigma_0^2\bm{I})$ with a standard deviation of $\sigma_0\!=\!\pi/3$. Then, we compute the norm of the simulated trajectory errors and validate that the exponential convergence bound \eqref{prob_theo1} holds for all realizations and times $t\!\geq\!0$. Here, we have used the parameter values $\varepsilon\!=\!1.1012$, $\vartheta\!=\!1.4211$ and minimal convergence rate $\underline{\alpha}\!=\!0.1056$, which were obtained by solving
\begin{equation}
\min_{\varepsilon,\vartheta,\underline{\alpha}} \underline{\varrho}(\varepsilon,\vartheta) \!+\! \tfrac{1}{\underline{\alpha}} \quad \mathrm{s.t.} \quad \kappa,\varphi \!>\! 0, \sqrt{\tfrac{\kappa}{\bar{\hat{m}}}} \!>\! \varepsilon \!>\! 0 ,\, \underline{\lambda}(\bm{\Upsilon}(\underline{\alpha})) \!\geq\! \underline{\upsilon} \notag
\end{equation}
numerically for $\underline{\upsilon}\!=\!6$ with the virtual stiffness $\kappa\!=\!21.12$ and scale $\varphi\!=\!0.3658$ following from \eqref{opt_kappa}--\eqref{opt_phi} and the worst-case radius $\underline{\varrho}$ from \eqref{th1_ball} for $\Delta\!=\!0.5269$. The time evolution of convergence rate and ball radius is illustrated in Fig~\ref{two_link_exp_cov_params} for the same L-GP nat-PD+ trajectories along with the variance-adaptive extension. The latter clearly shows an improved convergence behavior for the same initial conditions. 

Finally, in order to confirm the overall performance increase of the proposed structural preserving and variance-adaptive controllers, we perform a Monte Carlo simulation over the initial conditions and frequency $\omega$ of the sinusoidal reference $\bm{q}_d(t)\!\!=\!\!\pi/2\sin(\omega t)\bm{1}$. Therefore, we simulate 100 realizations per frequency sample, where the initial conditions of the two-link are drawn from a uniform distribution over $\bm{q},\dot{\bm{q}}\!\in\![-a,a]^4$ for $a\!=\!\pi/4$. The results in Fig.~\ref{tl_monte_carlo} show the $\mathcal{L}_2$-norms of the error state vector and actuation over the frequency $\omega$, evaluated in the steady-state for $t\!\geq\!4\pi/\omega$. Not only can we validate subsequent increases in the tracking accuracy, but we can also observe the superior robustness of the proposed controllers. Despite rising disturbance speeds, the natural PD+ controllers structurally ensure stability and, thus, a reliable performance, whereas the standard PD+ versions become unstable regardless of the used parametric or L-GP-based model. 

\begin{figure*}[]
\centering
\includegraphics[width=\textwidth]{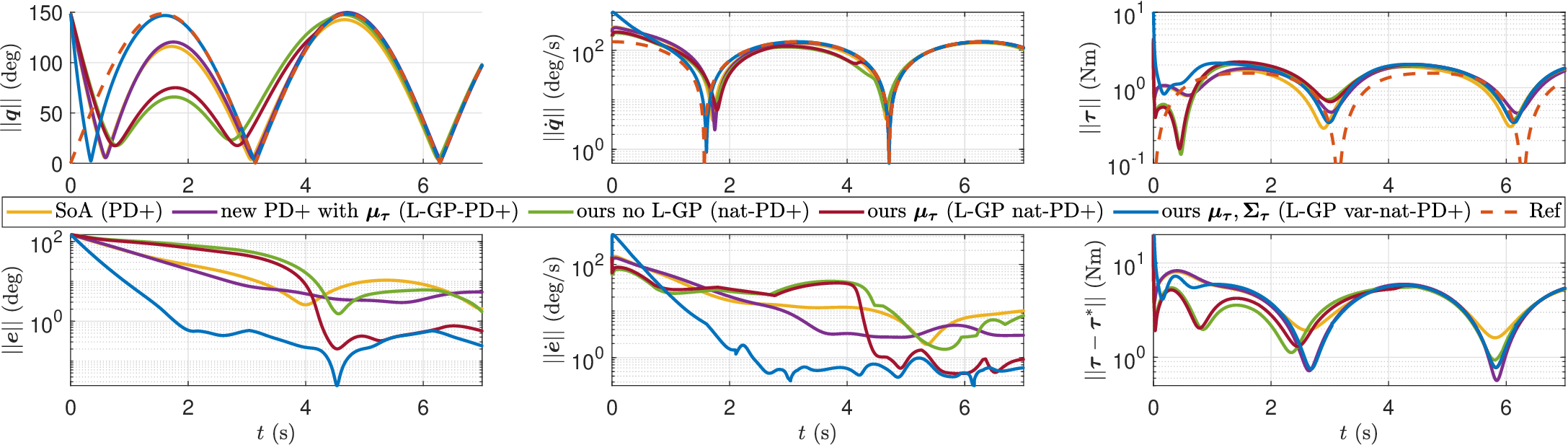}
\caption[]{Tracking performances of the standard parametric and proposed, natural structure-preserving or L-GP-based, PD+ controllers in the $N$-dimensional curvature space of the soft robot, $N\!\!=\!\!4$, for the initial condition $\bm{q}(0)\!\!=\!\!-\bar{\bm{a}}$, $\dot{\bm{q}}(0)\!\!=\!\!\bm{0}$. The reference torque $\bm{\tau}^*(t)$ indicates the nominal (feedforward) actuation based on the FEM for zero tracking errors.}% The depicted quantities are the relative minimum Euclidian joint position, velocity, acceleration, total input, and torque distances between the training and validation samples.}% The absolute minimal RMS differences w.r.t. $\bm{q}$, $\dot{\bm{q}}$, $\dot{\bm{q}}$ and $\bm{\tau}$ are , respectively.} 
\label{traj_and_errors_soro}
\end{figure*}

\subsection{Planar Soft Robot}

\subsubsection{Setup}

Next, we consider the planar soft robot from \cite{cosimo_ijrr_plan_soro,cosimo_soro_conf}, simulated by a FEM model. For this, we employ a discretization of a continuous rod with unit mass and length as a series of infinitesimal links \cite{penning_inf__discrete_rod}, where we consider a total of $N_{\text{FEM}}\!\!=\!\!100$ with lumped rotational principle inertias of $I_n\!\!=\!\!1/(12N_{\text{FEM}}^3)$, $n\!\!\in\!\!\{1,\dots,100\}$, each subsequently connected by linear torsional-spring-damper elements. These stiffnesses and dampings are set to $k_n\!\!=\!\!10$ Nm rad$^{-1}$ and $d_n\!\!=\!\!5$ Nm rad s$^{-1}$, respectively. The FEM simulation is implemented in Matlab based on the articulated body algorithm from \cite{featherstone_rbda_book}. As in Sec.~\ref{two_link_subsection} and \cite{cosimo_ijrr_plan_soro}, we rotate the base frame such that the soft arm is aligned with gravity in its equilibrium $\bm{q}^{\text{FEM}}\!=\!\bm{0}$ with $\bm{q}^{\text{FEM}}\!\in\!\mathbb{R}^{N_{\text{FEM}}}$.

For training, we consider the system's step response to a constant torque with amplitude $a\!=\!1$ Nm continuously acting on each FEM element, starting from equilibrium. The L-GP uses $D\!=\!24$ equidistant samples, corrupted by measurement noise with standard deviation $0.01a$, of the simulated trajectory for $0\!\!\leq\!\! t\!\!\leq\!\!4\!\!$~s, while discretizing into $N\!\!=\!\!4$ constant curvature (CC) segments~\cite{cosimo_ijrr_plan_soro} equivalent to a constrained rigid body with $4N\!\!=\!\!16\!\!$ DOFs. For the parameters, estimates $\hat{m}_n\!\!=\!\!(1\!+\!\chi_n)/N$, $\hat{k}_n\!\!=\!\!k_nN(1\!\!+\!\!\chi_n)/N_{\text{FEM}}$ and $\hat{d}_n\!\!=\!\!d_nN/N_{\text{FEM}}$ are used for the masses, stiffnesses and dampings, respectively, with relative errors $\chi_n\!\!=\!\!(-1)^{n-1}\!/4$ for $n\!\!\in\!\!\{1,2,3,4\}$. We reduce to Euclidian inertial and diagonal gravitational hypermetrics analogously to Sec.~\ref{two_link_subsection}. Additionally, we exploit symmetries of the robot's configuration by asserting symmetric \cite{duvenaud_2014} squared-exponential (SE) kernels for kinetic and potential energies, leading for the latter to a point-symmetric restoring potential force $\hat{g}(-\bm{q})\!\!=\!\!-\hat{g}(\bm{q})$. Then, the hyperparameters are optimized based on the least-squares approximation of the FEM positions with the L-GP-based dynamical resimulation, where we consider a sampling frequency of $250$~Hz leading to $V\!\!=\!\!1001$ validating position samples of the training trajectory.

\subsubsection{Control}

The standard and proposed, natural structure-preserving or L-GP-based, controllers with identical proportional and derivative gains $\bm{K}_P\!\!=\!\!\bm{K}_D\!\!=\!\!\bm{I}$ are tasked with tracking a sinusoidal reference $\bm{q}_d(t)\!\!=\!\!\bar{\bm{a}}\sin t$ with amplitude vector $\bar{\bm{a}}\!\!=\!\!\pi/180[1,10,45,90]$ in the considered $8$-dimensional CC state space, cf. Fig.~\ref{soro_ref_ws}. The system is initially assumed to be at rest, $\dot{\bm{q}}(0)\!\!=\!\!\bm{0}$, but deflected by the maximum amplitude $\bm{q}(0)\!\!=\!\!-\bar{\bm{a}}$. Projectors of the structure-preserving controllers are implemented again with the same $\epsilon\!\!=\!\!10^{-3}$ as in \eqref{project_epsilon}. The covariance gains \eqref{sigma_gain_adapt_law} of the adaptive natural PD+ controller are identically set to diagonal $\bm{K}_i\!\!=\!\!k_i\bm{I}$ with $k_1\!\!=\!\!10$, $k_2\!\!=\!\!10^{-3}$ and $k_3\!\!=\!\!10.05$ such that $\underline{k}_P\!\!=\!\!\underline{k}_D\!\!=\!\!1\!\!+\!\!0.1$. Also, instead of assuming their direct measurability, model-based estimates of the accelerations are used to compute the L-GP covariance.

\begin{table}[]
\begin{tabular}{l|ccccc}
& PD+    & L-GP-PD+ & nat-PD+ & \begin{tabular}[c]{@{}c@{}}L-GP\\ nat-PD+\end{tabular} & \begin{tabular}[c]{@{}c@{}}L-GP var-\\ nat-PD+\end{tabular} \\ \hline
$\lVert\bm{\tau}\rVert_{\mathcal{L}_2}$ & \textbf{5.00}  & 5.24 & 5.32 & 5.45 & 5.46 \\
$\max(\lVert\bm{\tau}\rVert)$\hspace{-0.1cm} & \textbf{1.89} & 1.96 & 1.99 & 2.04 & 2.05 \\
$\mathrm{E}[\lVert\bm{\tau}\rVert]$ & \textbf{1.313} & 1.395 & 1.408 & 1.436 & 1.439 \\ 
$\lVert[\bm{e}^\top\! \dot{\bm{e}}^\top\!]\rVert_{\mathcal{L}_2}$\hspace{-0.2cm} & 0.670 & 0.379 & 0.650 & 0.059 & \textbf{0.042} \\
$\max(\lVert\bm{e}\rVert)$  & 0.191 & 0.116 & 0.185 & 0.013 & \textbf{0.010} \\
$\max(\lVert\dot{\bm{e}}\rVert)$  & 0.181 & 0.089 & 0.210 & 0.022 & \textbf{0.017} \\
$\mathrm{E}[\lVert\bm{e}\rVert]$ & 0.123 & 0.088 & 0.118 & 0.008 & \textbf{0.005} \\
$\mathrm{E}[\lVert\dot{\bm{e}}\rVert]$ & 0.125 & 0.057 & 0.113 & 0.0135 & \textbf{0.0099}                                                      
\end{tabular}
\caption{Numerical evaluation of the steady-state control performances for $t\geq2\pi$s based on the soft robot. Lower values indicate better performances w.r.t.~the considered metrics.}
\label{table_num_vals_soro_ss}
\end{table}

\subsubsection{Results}

The controllers' tracking performances are visualized in Fig.~\ref{traj_and_errors_soro}, and their steady-state accuracies are evaluated numerically in Tab.~\ref{table_num_vals_soro_ss}. Aside from the significant performance increases achieved by the application of the L-GP model, the natural structure-preserving (L-GP nat-PD+) controller also leads to another drastic improvement by a factor of $6.47$ compared to the standard (L-GP-)PD+ variant w.r.t.~the $\mathcal{L}_2$-norm of the trajectory error. The proposed variance-adaptive (L-GP var-nat-PD+) controller demonstrates the highest tracking accuracy w.r.t.~all error metrics with slightly increased actuation effort, however.

\begin{figure}
\includegraphics[width=\linewidth]{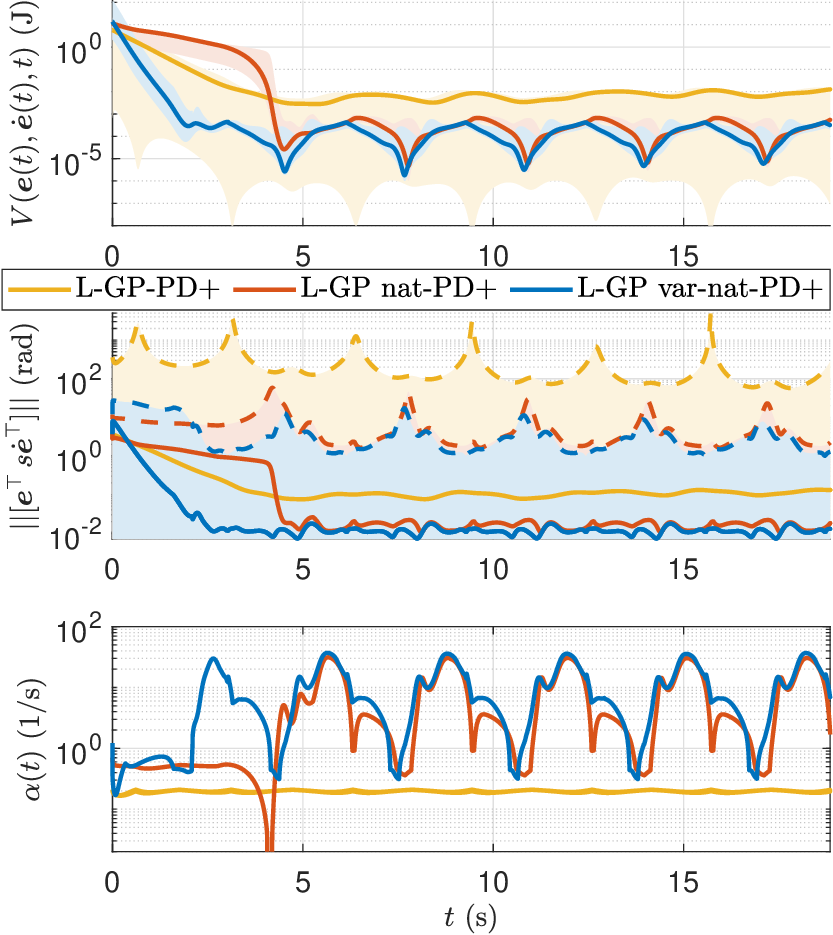}
\caption{Lyapunov functions, bounds, and exponential convergence rates of all L-GP-based controllers for the closed-loop trajectories from Fig.~\ref{traj_and_errors_soro}. Solid lines indicate trajectories, shaded areas functional \eqref{lmi_bounds_metric} or error norm \eqref{prob_theo1} bounds from Lemma~\ref{lemma_met} or Theorem~\ref{lgp_nat_pdp_theo1}, respectively.} 
\label{soro_nat_pdp_V_and_norm_ede}
\end{figure}

Next, we apply Theorem~\ref{lgp_nat_pdp_theo1} to compare the convergence and robustness of the controllers tracking the soft robot's reference curvature. The according Lyapunov functions \eqref{V_ct_relatship} and their bounds \eqref{lmi_bounds_metric} are evaluated in Fig.~\ref{soro_nat_pdp_V_and_norm_ede}, where the gravitational energy $\hat{G}(\bm{e})=\tfrac{1}{2}\bm{e}^\top\hat{\bm{\mathcal{K}}}_G(\bm{e})\bm{e}$ is eliminated from the metric~\eqref{coord_metric} with design stiffness \eqref{stiffness_metric} for the standard PD+ controller due to its cancellation. The norms of the closed-loop errors and the exponential convergence bounds \eqref{prob_theo1} are shown in the middle subplot of Fig.~\ref{soro_nat_pdp_V_and_norm_ede}, confirming their validity for each of the controllers $\forall t\!\geq\!0$ and also demonstrating their applicability to higher system dimensionalities. As in Sec.~\ref{two_link_subsub_results}, we optimally parametrize all guarantuees by leveraging \eqref{opt_kappa}--\eqref{opt_phi} and numerical minimization of the worst-case radius $\underline{\varrho}$ from \eqref{th1_ball}, for an L-GP model error bound of $\Delta\!=\!0.9556$. The resulting time evolutions of convergence rate $\alpha(t)$ and ball radius $\varrho(t)$ of all L-GP-based controllers are visualized in the bottom and middle subplots of Fig~\ref{soro_nat_pdp_V_and_norm_ede} by the solid and dashed lines, respectively. Compared to the standard (L-GP-)PD+, the structure-preserving (L-GP nat-PD+) controller without variance adaptation already leads to significantly increased robustness and convergence with an improved exponential rate and radius by up to two and three orders of magnitude, respectively. Finally, the variance-adaptive (L-GP var-nat-PD+) version shows the best behavior, with a big improvement compared to the static-gain (L-GP nat-PD+) counterpart, particularly visible in the transient phase $0\!\leq\!t\!\leq\! 5$~s.

\section{Conclusion}\label{sec_conclusion}

We have proposed a structure-preserving tracking controller using projections and a learned model given by the posterior of a Lagrangian-Gaussian Process (L-GP). The uncertainty quantification in the form of the covariance matrix is encoded into a confidence-dependent feedforward-feedback balancing scheme. High accuracy and performance are guaranteed by precise theoretical results ensuring exponential convergence and stability probabilistically. Numerical simulations of robotic systems, a planar two-link and soft manipulator, confirm the theoretical and practical efficacy of the proposed methods.

\appendices

\section*{Acknowledgment}

This work was supported by the Consolidator Grant ”Safe data-driven control for human-centric systems” (CO-MAN) of the European Research Council (ERC) by the European Union (EU) under grant agreement ID 864686.

\section{Supplementary Proofs}

\subsection{Proof of Lemma~\ref{lemma_met}}\label{l3_p_appendix}

Using Schur complements~\cite{boyd_lmis}, the equivalence
\begin{equation}
\bm{\mathcal{M}}\succ\bm{0} \quad \Leftrightarrow \quad \bm{\mathcal{K}} - \varepsilon^2\hat{\bm{M}}\hat{\bm{M}}^{-1}\hat{\bm{M}} \succ \bm{0} \notag
\end{equation}% given $\hat{\bm{M}} \succ \bm{0}$
can be used to follow that $\varepsilon^2\hat{\bm{M}}\prec\bm{\mathcal{K}}$ must hold for positive definiteness of $\bm{\mathcal{M}}$, since $\hat{\bm{M}} \succ \bm{0}$. Then, applying Weyl's inequalities \cite[III.2.1]{bhatia_ma}, we arrive at $\varepsilon^2\bar{\hat{m}}<\underline{\kappa}$. Taking the square root and reformulating leads to the first condition. Next, for the proof of \eqref{lmi_bounds_metric}, we start off with the nonstrict LMIs
\begin{equation}
\bm{0}\preceq\bm{\mathcal{M}} - \underline{\mu}\bm{I}, \quad \bm{0}\preceq \bar{\mu}\bm{I} - \bm{\mathcal{M}} \, . \notag
\end{equation}
Given $\hat{\bm{M}}\neq\bm{\mathcal{K}}$, we can always find positive values satisfying $\underline{\mu}<\max(\underline{\kappa},\underline{\hat{m}})$, $\bar{\mu}>\min(\bar{\kappa},\bar{\hat{m}})$ such that either $\bm{\mathcal{K}}-\underline{\mu}\bm{I},\bar{\mu}\bm{I}-\bm{\mathcal{K}}\succ\bm{0}$ or $\hat{\bm{M}}-\underline{\mu}\bm{I},\bar{\mu}\bm{I}-\hat{\bm{M}}\succ\bm{0}$ hold, respectively. Finally, let us consider the special case of $\hat{\bm{M}}=\bm{\mathcal{K}}$. Then, we can directly compute the eigenvalues of \eqref{coord_metric} analytically using its characteristic polynomial:
\begin{align}
\hat{\bm{M}}=\bm{\mathcal{K}} \,\,\Rightarrow\,\, &\det(\bm{\mathcal{M}}-\lambda\bm{I}) \notag\\
&= \det[(\hat{\bm{M}}-\lambda\bm{I})^2-(\varepsilon\hat{\bm{M}})^2] \notag \\
&= \det[(1\!-\!\varepsilon)\hat{\bm{M}}-\lambda\bm{I}]\det[(1\!+\!\varepsilon)\hat{\bm{M}}-\lambda\bm{I}] \notag \, ,
\end{align}
where we have applied the formula for determinants of 2$\times$2 block matrices \cite[Theorem 3]{silvester_block_det} in the second line, valid since $\varepsilon\hat{\bm{M}}$ and $\hat{\bm{M}}-\lambda\bm{I}$ commute, along with the multiplicativity property of the determinant in the last line. In this case, the spectrum of $\bm{\lambda}(\bm{\mathcal{M}})$ is therefore the result of the eigenvalues $\bm{\lambda}(\hat{\bm{M}})$ being scaled by $1\pm\varepsilon$, such that $\bm{\lambda}(\bm{\mathcal{M}})=(1\pm\varepsilon)\bm{\lambda}(\hat{\bm{M}})$ holds. Clearly, this coincides with $\underline{\bar{\mu}}=\underline{\bar{\lambda}}(\bm{\mathcal{M}})$ in \eqref{mu_bar_unli} for $\underline{\bar{\kappa}}=\underline{\bar{\hat{m}}}$, and positive definiteness necessarily requires $|\varepsilon|<1$ which is encompassed by the second (sufficient) condition of the lemma on $|\varepsilon|$ since $\underline{\hat{m}} / \bar{\hat{m}} \leq 1$.

\subsection{Proof of Corollary~\ref{coro_met}}\label{c1_p_appendix}

Using the commutability of $\varepsilon\hat{\bm{M}}$ and $\hat{\bm{M}}-\lambda\bm{I}$ combined with \cite[Theorem 3]{silvester_block_det}, we can compute the characteristic polynomial of \eqref{coord_metric} according to
\begin{align}
\bm{\mathcal{K}}=\kappa\bm{I}\,\,\Rightarrow\,\, &\det(\bm{\mathcal{M}}-\lambda\bm{I}) \notag\\
&= \det[(\bm{\mathcal{K}}-\lambda\bm{I})(\hat{\bm{M}}-\lambda\bm{I})-\varepsilon^2\hat{\bm{M}}^2] \notag \\
&= \det[\underbrace{(\kappa-\lambda)\hat{\bm{M}}-\varepsilon^2\hat{\bm{M}}^2}_{\coloneqq\hat{\bm{A}}}-\lambda(\kappa-\lambda)\bm{I}] \notag.
\end{align}
Thus, $\det(\bm{\mathcal{M}}-\lambda\bm{I})=0$ holds iff $\lambda(\kappa-\lambda)$ is an eigenvalue of $\hat{\bm{A}}$, leading for $\hat{\mu}\coloneqq\lambda(\hat{\bm{M}})$ to the algebraic condition
\begin{equation}
(\kappa-\lambda)\hat{\mu} - \varepsilon^2\hat{\mu}^2 = \lambda(\kappa-\lambda) \, , \label{quad_eig_val} 
\end{equation}
where we have used $\bm{A}^2\bm{x}=\lambda^2(\bm{A})\bm{x}$ for $\bm{A}\in\mathbb{R}^{N\times N}$. Finally, solving the quadratic equation \eqref{quad_eig_val} for the desired eigenvalues $\lambda\equiv\lambda(\bm{\mathcal{M}})$ and vectorizing, we arrive at \eqref{eig_vals_metric_eucl_K}.

\bibliographystyle{IEEEtran}

\bibliography{root}

\vspace{-3cm}
\begin{IEEEbiography}[{\includegraphics[width=1in,height=1.25in,clip,keepaspectratio]{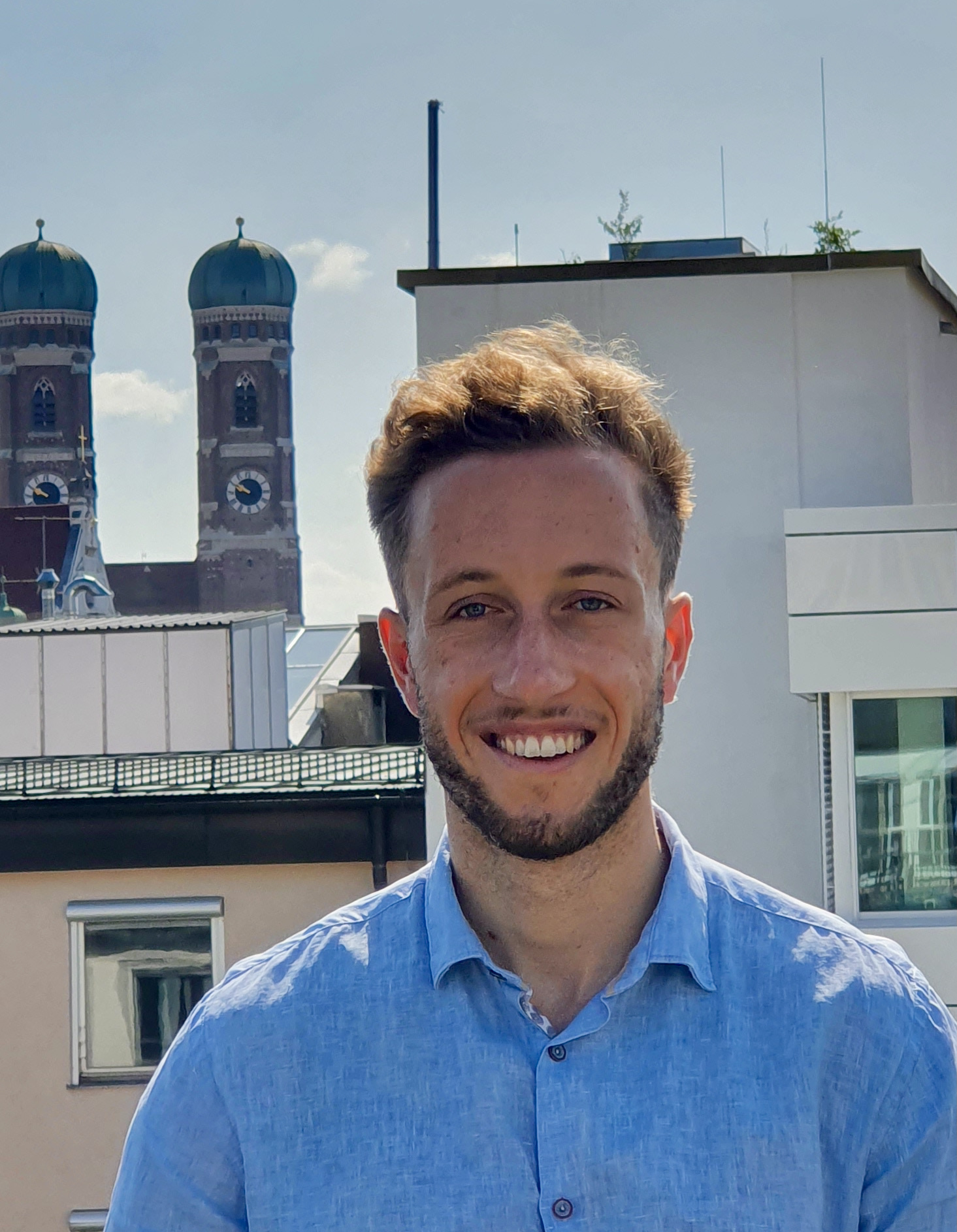}}]{Giulio Evangelisti}
received the B.Sc. and M.Sc. degrees in \fin{E}lectrical \fin{E}ngineering and \fin{I}nformation \fin{T}echnology from the Technical University of Munich (TUM), Germany, in 2017 and 2019, respectively. \fin{From 2017 to 2018, he was part of the signal generator department of the measurement technology division at Rohde \& Schwarz GmbH \& Co. KG, Munich, and from 2019 to 2020 a full-time control engineer at Blickfeld GmbH, Munich.} Since January 2021, he has been a Ph.D. \fin{candidate} at the Chair of Information-oriented Control, TUM School of Computation, Information and Technology. His current research interests include the stability of data-driven control systems, physically consistent machine learning, passivity-based control, and nonlinear systems.
\end{IEEEbiography}
\vspace{-3cm}
\begin{IEEEbiography}[{\includegraphics[width=1in,height=1.25in,clip,keepaspectratio]{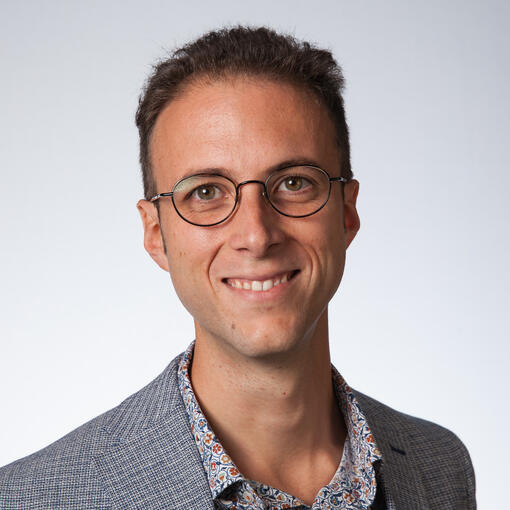}}]{Cosimo Della Santina}
\fin{is an assistant professor at TU Delft, 2628 CD Delft, The Netherlands, and a research scientist at the German Aerospace Institute. He received the Ph.D. degree in robotics (cum laude, 2019) from the University of Pisa. He was a visiting Ph.D. student and a postdoc (2017 to 2019) at the Computer Science and Artificial Intelligence Laboratory, Massachusetts Institute of Technology. He was a senior postdoc (2020) and guest lecturer (2021) in the Department of Informatics, Technical University of Munich. He has been awarded the euRobotics Georges Giralt Ph.D. Award (2020) and the IEEE Robotics and Automation Society Fabrizio Flacco Young Author Award (2019), and he has been a finalist for the European Embedded Control Institute Ph.D. Award (2020). In 2023, he received the IEEE RAS Early Academic Career Award in Robotics and Automation. His research interests include the motor intelligence of physical systems, with a focus on elastic and soft robots.}
\end{IEEEbiography}
\vspace{-3cm}
\begin{IEEEbiography}[{\includegraphics[width=1in,height=1.25in,clip,keepaspectratio]{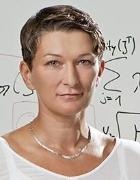}}]{Sandra Hirche}
\fin{holds the TUM Liesel Beckmann Distinguished Professorship and heads the Chair of Information-oriented Control in the Faculty of Electrical and Computer Engineering at Technical University of Munich (TUM), Germany (since 2013). She received the diploma engineer degree in Aeronautical and Aerospace Engineering in 2002 from the Technical University Berlin, Germany, and the Doctor of Engineering degree in Electrical and Computer Engineering in 2005 from the Technische Universität München, Munich, Germany. From 2005-2007 she has been a PostDoc Fellow of the Japanese Society for the Promotion of Science at the Fujita Laboratory at Tokyo Institute of Technology, Japan. Prior to her present appointment she has been an Associate Professor at TUM. Her main research interests include  learning, cooperative, and distributed control with application in human-robot interaction, multi-robot systems, and general robotics. She has published more than 200 papers in international journals, books, and refereed conferences.}
\end{IEEEbiography}

\end{document}